\documentclass[lettersize,journal]{IEEEtran}

\usepackage{booktabs}
\usepackage[ruled]{algorithm2e}
\SetAlFnt{\small}
\SetAlCapFnt{\small}
\SetAlCapNameFnt{\small}
\SetAlCapHSkip{0pt}
\IncMargin{-\parindent}

\usepackage{amsmath,amssymb,amsfonts}
\usepackage{comment}
\allowdisplaybreaks

\usepackage{latexsym}
\usepackage{siunitx}

\usepackage{bm}
\usepackage{nicefrac}
\usepackage{booktabs}
\usepackage{array}
\usepackage{multirow}
\usepackage[para,online,flushleft]{threeparttable}
\usepackage{makecell}
\usepackage{stfloats}
\usepackage{graphicx}
\usepackage{color}
\usepackage{url}

\usepackage{tabularx}
\usepackage{multirow}
\usepackage{makecell}
\usepackage{siunitx}
\usepackage{subfigure}

\newtheorem{theorem}{Theorem}[section]
\newtheorem{lemma}[theorem]{Lemma}
\newtheorem{definition}[theorem]{Definition}
\def\proof{{\bf Proof.}\hskip 0.3truecm}
\def\endproof{\quad $\Box$}

\allowdisplaybreaks[3]

\usepackage[caption=false,font=normalsize,labelfont=sf,textfont=sf]{subfig}
\usepackage{textcomp}
\usepackage{url}
\usepackage{verbatim}
\usepackage{cite}
\DontPrintSemicolon
\SetKw{KwAnd}{and}
\SetFuncSty{textsc}
\SetKwInOut{Input}{Input\ \ \ \ }
\SetKwInOut{Output}{Output}

\newenvironment{envblue}{}{}

\newcommand{\smblue}[1]{{#1}}
\newcommand\ppi{\boldsymbol{\pi}}
\newcommand\ppsi{\boldsymbol{\mathit{\psi}}}

\newcommand\ee{\boldsymbol{\mathit{e}}}

\newcommand\qq{\boldsymbol{\mathit{q}}}

\newcommand\xx{\boldsymbol{\mathit{x}}}
\newcommand\yy{\boldsymbol{\mathit{y}}}
\newcommand\zz{\boldsymbol{\mathit{z}}}

\renewcommand\AA{\boldsymbol{\mathit{A}}}
\newcommand\BB{\boldsymbol{\mathit{B}}}
\newcommand\DD{\boldsymbol{\mathit{D}}}

\newcommand\II{\boldsymbol{\mathit{I}}}
\newcommand\JJ{\boldsymbol{\mathit{J}}}
\newcommand\LL{\boldsymbol{\mathit{L}}}
\newcommand\calLL{\boldsymbol{\mathcal{L}}}

\newcommand\PP{\boldsymbol{\mathit{P}}}
\newcommand\QQ{\boldsymbol{\mathit{Q}}}

\renewcommand\SS{\boldsymbol{\mathit{S}}}
\newcommand\WW{\boldsymbol{\mathit{W}}}
\newcommand\ZZ{\boldsymbol{\mathit{Z}}}

\newcommand\Otil{\widetilde{O}}
\newcommand\ZZtil{\boldsymbol{\mathit{\tilde{Z}}}}

\newcommand\LLbar{\boldsymbol{\mathit{\bar{L}}}}

\newcommand{\var}[1]{\textbf{var} \left\{ #1 \right\} }

\newcommand{\E}{\mathcal{E}}
\newcommand{\V}{\mathcal{V}}

\newcommand{\N}{{\mathcal N}}

\newcommand{\G}{\mathcal{G}}

\newcommand{\kh}[1]{\left(#1\right)}
\def\abs#1{\left|#1  \right|}
\def\norm#1{\left\| #1 \right\|}
\def\calG{\mathcal{G}}
\def\G{\mathcal{G}}
\def\eps{\epsilon}
\newcommand{\one}{\mathbf{1}}

\def\defeq{\stackrel{\mathrm{def}}{=}}

\newcommand{\mat}[1]{\boldsymbol{#1}}
\renewcommand{\vec}[1]{\boldsymbol{#1}}

\newcommand{\mD}{\ensuremath{\mat{D}}\xspace}

\newcommand{\mI}{\ensuremath{\mat{I}}\xspace}

\newcommand{\mP}{\ensuremath{\mat{P}}\xspace}

\newcommand{\mS}{\ensuremath{\mat{S}}\xspace}

\newcommand{\ve}{\ensuremath{\vec{e}}\xspace}

\newcommand{\vp}{\ensuremath{\vec{p}}\xspace}

\newcommand{\vw}{\ensuremath{\vec{w}}\xspace}

\newcommand{\bbE}{\ensuremath{{\mathbb E}}\xspace}

\newcommand{\bbP}{\ensuremath{{\mathbb P}}\xspace}

\newcommand{\cV}{\ensuremath{{\mathcal V}}\xspace}

\newcommand{\cX}{\ensuremath{{\mathcal X}}\xspace}

\begin{document}
\title{Behavior and Sublinear Algorithm for Opinion Disagreement on Noisy Social Networks}
\author{Wanyue~Xu, Yubo~Sun, Mingzhe~Zhu, Zuobai~Zhang, and  Zhongzhi~Zhang~\IEEEmembership{Member,~IEEE}

    \thanks{Wanyue Xu, Yubo Sun, Mingzhe Zhu, Zuobai Zhang, and Zhongzhi Zhang are with the College of Computer Science and Artificial Intelligence, Fudan University, Shanghai 200433, China. Wanyue Xu is also with Shanghai Key Laboratory of Brain-Machine Intelligence for Information Behavior, Shanghai 201620, China; and the School of Business and Management, Shanghai International Studies University, Shanghai 201620, China.
  
        xuwy@fudan.edu.cn; 25110890019@m.fudan.edu.cn; 21210240107@m.fudan.edu.cn; 17300240035@fudan.edu.cn; zhangzz@fudan.edu.cn.}

}

\markboth{}%
{Xu \MakeLowercase{\textit{et al.}}: Behavior and Sublinear Algorithm for Opinion Disagreement on Noisy Social Networks}


\maketitle

\begin{abstract}
The phenomenon of opinion disagreement has been empirically observed and reported in the literature, which is affected by various factors, such as the structure of social networks. An important discovery in network science is that most real-life networks, including social networks, are scale-free and sparse. In this paper, we study noisy opinion dynamics in sparse scale-free social networks to uncover the influence of power-law topology on opinion disagreement. We adopt the popular discrete-time DeGroot model for opinion dynamics in a graph, where nodes' opinions are subject to white noise.  We first study opinion disagreement in many realistic and model networks with a scale-free topology, which approaches a constant, indicating that a scale-free structure is resistant to noise in the opinion dynamics. Moreover, existing algorithms for estimating opinion disagreement are computationally impractical for large-scale networks due to their high computational complexity. To solve this challenge, we introduce a sublinear-time algorithm to approximate this quantity with a theoretically guaranteed error. This algorithm efficiently simulates truncated random walks starting from a subset of nodes while preserving accurate estimation. Extensive experiments demonstrate its efficiency,  accuracy, and scalability.
\end{abstract}


\section{Introduction}
\IEEEPARstart{W}
ith the rapid development of Internet technology in recent years, social media and online social networks have experienced explosive growth~\cite{Le20}, which constitute an indispensable part of people's lives~\cite{SmCh08}. Due to the absence of geographical barriers erased by the Internet, social media and on-line social networks have created many new ways for people across the globe to exchange in real time their opinions or information about some topics or issues, resulting in a substantial change of ways people share and form opinions~\cite{AnYe19}. At the same time, the huge prevalence of online social media platforms also leads to various complex social phenomena, such as polarization~\cite{MaTeTs17,MuMuTs18,ZhBaZh21,HaMeRiUp21}, disagreement~\cite{GaKlTa20,WaZhZh25,SuSuZhZh25,QiQiZhLi26}, and diversity~\cite{MaPa19}, all of which have become a hot topic of study in the literature of social science~\cite{AxDaFo21}, control science~\cite{YiPa20}, and computer science~\cite{XuBaZh21,ZhBaZh21,SaKeKhLa23}.

For better understanding the diffusion, evolution, and formation of opinions, as well as the social phenomena created by online social media, a rich variety of models for opinion dynamics have been developed~\cite{No20}. These models have been applied to various practical scenarios~\cite{BeWaVaHoShAl21,FrPrTePa16,FrBu17}, such as accounting for the Paris Agreement negotiation process~\cite{BeWaVaHoShAl21} and  the change of public opinion on the US-Iraq War~\cite{FrBu17}.
Among diverse opinion dynamics models, the DeGroot model~\cite{De74} is a popular one. 
In this model, each node has a real scalar-valued initial opinion. At each time step, every node updates its opinion by averaging the opinions of all its neighbors at the previous step. Finally, the opinions of all nodes reach consensus. Since its establishment, the DeGroot model has been modified or generalized, to incorporate various factors contributing to opinion dynamics~\cite{No20}, such as noise~\cite{XiBoKi07,JaOl19,XuZh23CIKM}. In opinion dynamics, there often exist stochastic changes for exogenous factors, such as uncertainty of communication environment and linguistic opinion, which can be represented as noise. In~\cite{XiBoKi07} a significant extension of the DeGroot model was proposed by including noise. For this noisy model, the opinions of nodes will not reach agreement but fluctuate around a weighted average of their opinions in long time limit, leading to opinion disagreement/diversity/discord, which is quantified by a weighted squared deviation of each node's opinion from their average.

The opinion disagreement in the noisy DeGroot model exhibits rich behavior dependent on the topology of interaction systems~\cite{JaOl19}. For example, opinion disagreement behaves as a different function of the node number in the path graph, the sparse regular graphs, and the complete binary tree. However, these graphs cannot well mimic realistic networks. It has been established~\cite{BaAl99} that most real-world networks including social networks are sparse having a small constant average degree, and display the scale-free topology, with their node degrees $k$ following a power-law distribution $P(k)\sim k^{-\gamma}$ with $2<\gamma\leqslant 3$, which has attracted much interest from the community of graph mining~\cite{WeHeXiWaLiDuWe19,XuShZhKaZh20,MaJa21,XuZh23}. The scale-free structure has a strong influence on various dynamical processes taking place in power-law graphs~\cite{Ne03}. However, the effect of scale-free topology on opinion disagreement for the noisy DeGroot model is still not well understood. Moreover, direct exact computation of opinion disagreement needs cubic time, which is time-consuming for large networks. These motivate us to explore the effect of scale-free structure on opinion disagreement and to develop efficient algorithms to solve the challenge for computing opinion disagreement.

\begin{envblue}
    To overcome the challenge of directly computing opinion disagreement, we develop two novel approximation algorithms with theoretical guarantees. The first algorithm achieves nearly linear time complexity by combining spectral sparsification techniques with Laplacian solvers and Johnson-Lindenstrauss lemma, carefully addressing self-loops in the graph transformation. The second algorithm takes a different approach through sublinear-time random walk samplings. It estimates returning probabilities of nodes via truncated walks from a carefully selected subset of nodes, avoiding matrix operations entirely. Numerical results validate its scalability and provable error bounds for networks with tens of millions of nodes.
\end{envblue}

Our main contributions are summarized as follows.
    
(i)    We derive an analytical expression for the opinion disagreement on an arbitrary connected weighted non-bipartite graph $\G$, in terms of the eigenvalues and eigenvectors of the normalized adjacency matrix associated with $\G$, which depend on the structural properties of $\G$.
 
(ii)  Based on the obtained expression, we examine the effect of scale-free topology on opinion disagreement by studying the noisy DeGroot model on various real and model scale-free networks, which tends to small constants, independent of the network size, indicating that scale-free structure is resistant to noise in the DeGroot model for opinion dynamics.

(iii)   We develop two approximation algorithms to estimate opinion disagreement, which have theoretically guaranteed errors. The first one has nearly linear time complexity with respect to the number of edges in $\G$, while the second one is sublinear with respect to the node number. In contrast to direct exact computation of opinion disagreement, our approximation algorithms significantly reduces the time cost.

(iv) We perform extensive experiments on various real-world networks and model networks, which demonstrate that our proposed approximation algorithms are both efficient and effective, with the second algorithm being  scalable to large networks with over fifty millions of nodes.

\section{Related Work}

This section is devoted to brief review of some previous work related to ours.

Among different models for opinion dynamics, the DeGroot model~\cite{De74} is a popular one. In this model, every individual/agent has only one scalar-valued opinion, which evolves as an average of its neighbors' opinions. A necessary and sufficient condition for reaching consensus based on the DeGroot model was obtained in~\cite{Be81}. Although simple and succinct, the DeGroot model captures the role of interactions among individuals in the opinion evolution. However, the factors contributing to the propagation, evolution, and formation for opinions are diverse, most of which are overlooked in the original DeGroot model. In order to better model opinion dynamics, in the past years, numerous modifications or extensions of the DeGroot model have been proposed by incorporating different aspects affecting the behavior or distribution of opinions~\cite{No20}. 

An important extension of the DeGroot model is the Friedkin-Johnsen (FJ) model~\cite{FrJo90}. It introduces a resistance parameter for each agent to the DeGroot model, reflecting the agent's willingness to conform with its neighbors' opinions. In the FJ model, each agent possesses two opinions, innate opinion and expressed opinion, with the long-run expressed opinions converging but never reaching agreement. The FJ model has been applied to analyze various social phenomena~\cite{ChMu20,TuNe22}. Another major variant of the DeGroot model is the Altafini model~\cite{Al13}, where the interactions between agents can be either cooperative or antagonistic~\cite{ZhSuXuLiZh24}. In~\cite{ZhWuAl20}, a modification of the DeGroot model is studied, where an agent updates its opinion by taking into account both its neighbors' opinions and the opinions of its neighbors' neighbors~\cite{ZhXuZhCh24}. Moreover, existing work also extends the DeGroot model by incorporating leaders, whose opinions remain unchanged. These DeGroot models with leaders are formulated in the framework of optimization, such as maximizing the diversity~\cite{MaPa19} and the total opinion~\cite{MeAsDaAmAn13,VaFaFr14,MaAb19,ZhZh21,ZhZhLiZh23,ZhZh23}. Finally, other variants or modifications of the DeGroot model were also proposed and studied~\cite{DaGoLe13,DoDiMa17}.


The above-mentioned models provide insights into understanding opinion dynamics, since they capture some important aspects contributing opinion formation, including agent's attributes, interactions among agents, and opinion updating rules. However, in the majority of currently existing models for opinion dynamics based on the DeGroot model, opinion update rules are deterministic. In practical situations, individual opinions are not governed deterministically but are subject to stochastic change from a multitude of factors, for example, exogenous factors unrelated to social influence. In order to capture exogenous noise, a noisy DeGroot model was introduced and studied~\cite{XiBoKi07}, which was later extended to the FJ model~\cite{StLi20}. In the noisy model, the limiting opinions do not converge but fluctuate around their weighted average, leading to opinion disagreement $\delta$ quantified by the expected mean-squared deviation of opinions, the upper and lower bounds of which were studied in~\cite{LoGaZa13}.

\begin{envblue}
    While our work focuses on the noisy DeGroot model to isolate the impact of scale-free topology, several extensions of this framework merit discussion. Prior studies have proposed enhancements to opinion dynamics models that incorporate cognitive biases~\cite{So18,ChWaTs22}, heterogeneous susceptibility to influence~\cite{BeVaIe22,MeRaBrHaRo24,AbKlPaTs18,LiYeAnBaNe18}, and alternative formulations of uncertainty~\cite{Ch18,ZhLiKoDoYu19,FoKaKoSk18}. These factors better capture behavioral nuances in real social systems, which also deepen our understanding of structural robustness. Dynamic network paradigms, where topology evolves with time and opinions~\cite{HeFaZhWa23}, represent another critical frontier for modeling real-world social networks. Notably, signed networks where edges encode trust or distrust relationships~\cite{ArPaSa23,HeFeChLiHuTa22,ArPa24}, also offer a natrual extension for modeling opinion dynamics~\cite{HeSuWaWaHuYiWaMa21}. However, these sophisticated models often sacrifice analytical tractability for ecological validity on large-scale networks. Our choice to employ the homogeneous noisy DeGroot model creates a controlled baseline that rigorously quantifies purely topological effects in power-law networks. This foundational analysis enables future work to systematically introduce additional layers of complexity. Meanwhile, it also retains the ability to disentangle structural influences from behavioral factors.
\end{envblue}

The opinion disagreement $\delta$ and its leading behavior in some typical networks with various structural properties have been deeply studied in~\cite{JaOl19}, which shows that the scaling behavior for $\delta$ in different networks is rich. For a network with $N$ nodes, $\delta$ can behave linearly, logarithmically, or independently of $N$.
For example, in the path graph, the ring graph, and sparse regular graphs, $\delta$ scales linearly with $N$; in the complete binary tree and the 2-dimensional grid, $\delta$ grows logarithmically with $N$; while in the star graph, the complete graph, regular $\alpha$-expander graphs, regular dense graphs, dense Erd\"os-R\'enyi random graphs, and $d$-dimensional grids with $d\geqslant 3$, $\delta$ are constants, irrespective of $N$. Therefore, the behavior of opinion disagreement $\delta$ is dependent on network structure and average node degree. Thus far, the impact of scale-free structure on $\delta$ has not been explored, despite the ubiquity of scale-free networks in real-world social and natural systems~\cite{BaAl99,Ne03}.
While the computation~\cite{ChLiDe18}, property~\cite{GaKlTa20} and optimization~\cite{MuMuTs18} of $\delta$ have been investigated, existing literature predominantly focuses on the FJ model, whose formulation differs from the noisy DeGroot model. Furthermore, previous evaluations have been restricted to graphs with less than 200 thousand edges, leaving the scalability issue for large-scale networks largely unresolved.

\section{Preliminaries}

In this section, we give a brief introduction to some essential
concepts about a graph, its associated matrices, random walks and their relevant quantities on a graph.

\subsection{Graph and Matrix Notation}

Let $\G=(\V,\E,w)$ be a connected undirected weighted network with or without self-loops, where $\V=\{1,2,...,N\}$ denotes the set of $N$ nodes, $\E$ denotes the edge set consisting of $M$ unordered node pairs $\{i,j\}$, $i,j\in \V$, and $w:\E\rightarrow\mathbb{R}_{+}$ is the positive edge weight function, with $w(e)=w_e$ denoting the weight for edge $e$. The adjacency matrix $\AA$ of graph $\G$ is an $N\times N$ symmetric matrix, which encodes many topological and weighted properties of the graph, with the entry $a_{ij}$ at row $i$ and column $j$ defined as follows: $a_{ij}=w_e$ if nodes $i$ and $j$ are linked by an edge $e\in\E$, $a_{ij}=0$ otherwise. Let $\N_i$ denote the set of neighbors for node $i$, then the degree of $i$ is $d_i=\sum_{j\in\mathcal{N}_i}a_{ij}$. And the total degree $d_{\rm sum}$, the sum of degrees over all the $N$ nodes, is $d_{\rm sum}=\sum_{i=1}^{N} d_i$. The degree matrix $\DD$ of graph $\G$ is a diagonal matrix whose $i$th diagonal entry is $d_i$. And the Laplacian of $\G$ is defined as $\LL=\DD-\AA$, which is positive semi-definite satisfying $\LL\one=\mathbf{0}$, where $\one$ and $\mathbf{0}$ represent the $N$-dimensional vectors with all entries being ones and zeros, respectively.

The transition matrix $\PP$ of a graph $\G$ is $\PP=\DD^{-1}\AA$, which is row-stochastic.
Generally, $\PP$ is asymmetric with the exception of the case that $\G$ is a regular graph.
However, $\PP$ is similar to a real symmetric matrix $\SS$, called normalized adjacency matrix of $\G$ and defined by
$\SS = \DD^{-1/2}\AA\DD^{-1/2} = \DD^{1/2}\PP\DD^{-1/2}$.
Thus, $\PP$ and $\SS$ have identical set of eigenvalues. Let $\II$ stand for the identity matrix of dimension $N$.
Then, the normalized Laplacian matrix of graph $\G$~\cite{Ch97} is defined as $\calLL=\II-\SS$. 

\subsection{Random Walks on a Graph}

For a connected graph $\G$, we can define a discrete-time random walk running on it.
At each time step, the walker starting from its current location $i$ jumps to a neighboring node $j$ with probability $a_{ij}/d_i$. Such a stochastic process is described by a Markov chain~\cite{KeSn76}, characterized by the transition matrix $\PP$ defined above, with the $ij$-th entry $p_{ij}=a_{ij}/d_i$ representing the probability of moving from $i$ to $j$ in one time step. For a non-bipartite graph $\G$, the random walk has a unique stationary distribution $\ppi=(\pi_1, \pi_2, \cdots, \pi_N)^{\top}$, with $\pi_i=d_i/d_{\rm sum}$ satisfying $\sum_{i=1}^N \pi_i=1$ and $\ppi^\top\PP=\ppi^\top$. Note that for a transition matrix $\PP$, its square $\PP^2$ exhibits the same row-stochastic property and stationary distribution \(\ppi\) as \(\PP\):
\begin{envblue}
    \begin{equation}\label{eq:row-stochastic-sqtrans}
        \PP^2\one=\PP\one=\one,
    \end{equation}
    \begin{equation}\label{eq:stat-distr-sqtrans}
        \ppi^\top\PP^2=\ppi^\top\PP=\ppi^\top.
    \end{equation}
\end{envblue}
As will be shown later, based on graph $\G$ one can construct a new graph $\G'$ with transition matrix $\PP^2$. In this paper, we only consider non-bipartite graphs.
\begin{envblue}
    Table~\ref{tab:notation} lists the frequently used notations of different types of graphs related to \(\G\) throughout this paper.
\end{envblue}
\begin{table}[htbp]
    \centering
    \caption{Frequently used notations.}
    \label{tab:notation}
    \begin{tabularx}{\linewidth}{rX}
        \toprule
        \textbf{Notation} & \textbf{Description}                                    \\
        \midrule
        \(\G=(\V,\E,w)\)  & A connected undirected weighted network.                \\
        \(\PP\)           & The transition matrix of random walks on \(\G\).        \\
        \(\G'\)           & The weighted graph with transition matrix \(\PP^2\).    \\
        \(\tilde{\G}\)    & The sparsified graph of \(\G'\) by Lemma~\ref{eq:Ltil}. \\
         \(\bar{\G}\)    & The graph obtained from  \(\tilde{\G}\) by deleting loops. \\
        \bottomrule
    \end{tabularx}
\end{table}
A fundamental quantity related to random walks is hitting time, also called first-passage time~\cite{Lo93,CoBeTeVoKl07}.
For a random walk on graph $\G$ with transition matrix $\PP$, the hitting time from node $i$ to node $j$ is the expected number of jumps for a walker starting from $i$ to visit $j$ for the first time, denoted by $H_{ij}(\G)$ or $H_{ij}(\PP)$.
Hitting times have found wide applications in different areas~\cite{XuShZhKaZh20}. For example, the partial mean hitting time $H_j(\G)$~\cite{Be09,TeBeVo09} to a node $j$, can be used as a centrality measure of node $j$ in graph $\G$~\cite{MaMagi15,ZhXuZh20,XiXuZhZh25}. This quantity also plays an important role in noisy opinion dynamics~\cite{JaOl19}.
\begin{envblue}
    \begin{definition}
        The partial mean hitting time \(H_j(\G)\) is defined as the expected time for a walker arriving at node $j$, which starts from a node $i$ selected randomly according to the stationary distribution, that is,
        \begin{equation}\label{eq:def-hit-centr}
            H_j(\G)=\sum\nolimits_{i=1}^{N}\pi_i H_{ij}(\G).
        \end{equation}
    \end{definition}

\end{envblue}


\section{Discrete-Time Noisy DeGroot Model}

In this section, we introduce the noisy discrete-time DeGroot model for opinion dynamics. 

\subsection{Noisy DeGroot model}

We first introduce the popular discrete-time DeGroot model~\cite{De74} on a graph $\G$, where each node has a nonnegative scalar-value opinion. Let $\xx(t) \in \mathbb{R}^N$ denote the vector of node opinions at time $t$, with the entry $x_i(t)$ being the opinion of node $i$. At time $t+1$, node $i$ updates its opinion by averaging the opinions of all its neighbours at time $t$ as
{\small{$x_i(t+1) = \frac{\sum_{j \in \N_i}a_{ij}x_j(t)}{\sum_{j \in \N_i}a_{ij}}=\sum_{j \in \N_i}p_{ij}x_j(t)$}},
The evolution of node opinions can be represented in matrix form as $\xx(t+1)=\PP\xx(t)$.
After long-time evolution, the opinion vector $\xx(t)$ converges to $\big(\sum_{i=1}^N \pi_ix_i(0)\big)\one$, where $x_i(0)$ is the initial opinion of node $i$. Thus, all nodes reach a consensus value $\sum_{i=1}^N \pi_ix_i(0)$.

In~\cite{XiBoKi07}, the discrete-time DeGroot model is extended by incorporating noise. In this case, the opinion vector evolves according to the following rule $\xx(t+1) = \PP\xx(t) + \bm{\phi}(t)$,
where vector $\bm{\phi}(t)=(\phi_1(t), \phi_2(t), \cdots, \phi_N(t))^{\top}$ represents zero-mean independent identically distributed noise. In this paper, we focus on the impact of network topology on the behavior of opinions. We assume that the variance of the noise $\phi_i(t)$ at each node $i$ is unit, and that
noises at different nodes are uncorrelated. 

In the presence of noise, the opinion vector $\xx(t)$ does not converge, with the opinions of all nodes fluctuating around the weighted average of their current opinions when $t$ is large. These fluctuations can be characterized by the concept of opinion disagreement defined below.
\begin{definition}
    The opinion disagreement for discrete-time noisy DeGroot model for opinion dynamics is defined as the weighted average of the variance of the deviation from the weighted average of current opinions for nodes, with both averages being based on the stationary distribution $\ppi$:
    \begin{equation}\label{eq:delta0}
        \delta\defeq \lim_{t\to \infty}\sup \sum\nolimits_{i=1}^{N} \pi_i\,\var{x_i(t)-\sum\nolimits_{j=1}^{N}\pi_j x_j(t)}.
    \end{equation}
\end{definition}
It has been shown~\cite{JaOl19} that the opinion disagreement depends on the hitting times of random walks associated with transition matrix $\PP^2$ as
\begin{equation}\label{eq:delta}
\delta=\sum\nolimits_{i=1}^{N}\pi_i^2\sum\nolimits_{j=1}^{N}\pi_j H_{ji}(\PP^2).
\end{equation}
Equation~\eqref{eq:delta} explicitly reveals the relationship between disagreement and the hitting time: a larger hitting time makes local noise harder to average out, thereby amplifying the steady-state opinion variance.

\subsection{Spectral Expression for Opinion Disagreement}\label{sec:formula}

In fact, the transition matrix $\PP^2$ is associated with a connected weighted graph $\G'$ with loops, which can be constructed from $\G$ as follows~\cite{ChChLiPeTe15b}. The node set of $\G'$ is $\V=\{1,2,...,N\}$, identical to that of $\G$. For the adjacency matrix of $\G'$, its $ij$th entry can be obtained by performing a two-step random walk on $\G$. Concretely, two nodes $i$ and $j$ in $\G'$ are directly linked by an edge $e'$ if and only if there is a $2$-length path between $i$ and $j$ in $\G$. And the weight of edge $e'$ is equal to the product of the degree $d_i$ for node $i$ in $\G$ and the probability that a walker starts from $i$ and ends at $j$ after performing two-step random walks in $\G$.  It is easy to verify that the degree of each node $i$ in $\G'$ is $d_i$, identical to that in $\G$:
\begin{envblue}
    \begin{align*}
\DD\PP^2\one=\DD\one=\kh{d_1,d_2,\dots,d_n}^\top,
    \end{align*}
    where the first equation is due to~\eqref{eq:row-stochastic-sqtrans}.
\end{envblue}
Thus, the degree matrix of graph $\G'$ is also $\DD$. Then, for graph $\G'$, the normalized Laplacian matrix $\calLL(\G')$ is
$\calLL(\G')=\II-\SS^2
    =\DD^{1/2}(\II-\PP^2)\DD^{-1/2}$,
and the Laplacian matrix $\LL(\G')$ is
\begin{equation}\label{Lap}
    \LL(\G')=\DD^{1/2}\calLL(\G')\DD^{1/2}=\DD(\II-\PP^2).
\end{equation}
\smblue{Combining~\eqref{eq:delta} with~\eqref{eq:def-hit-centr}}, opinion disagreement $\delta$ can be rewritten as
\begin{equation}\label{eq:delta02x}
    \delta=\sum \nolimits_{i=1}^{N}\pi_i^2H_{i}(\G').
\end{equation}

Let $\lambda_1,\lambda_2,\ldots,\lambda_N$ denote the $N$ eigenvalues of matrix $\SS$, and let $\ppsi_1,\ppsi_2,\ldots,\ppsi_N$ be their corresponding mutually orthogonal unit eigenvectors.
Since all $\lambda_i$ are real, one can rearrange them in a decreasing order as $1=\lambda_1>\lambda_2\ge\lambda_3\ge\cdots\ge\lambda_N\ge-1$ and let $\lambda = \max\{\abs{\lambda_2}, \abs{\lambda_N}\}$. Thus, for matrix $\SS^2$ that is similar to its transition matrix $\PP^2$, its eigenvalues are $\lambda_i^2$ with corresponding eigenvectors $\ppsi_i$, $i=1,2,\ldots,N$. Then, the pseudoinverse or Moore-Penrose generalized inverse~\cite{BeGrTh74} of $\calLL(\G')$, which we denote $\calLL^\dagger(\G')$, can be written as
$\calLL^\dagger(\G')= \sum_{k=2}^{N}\frac{1}{1-\lambda_i^2} \ppsi_k\ppsi_k^\top$. Thus, the $i$th diagonal entry $\calLL_{ii}^\dagger(\G')$ of $\calLL^\dagger(\G')$ is $\sum_{k=2}^{N}\frac{1}{1-\lambda_k^2}\psi^2_{ki}$.


Using a way similar to that in~\cite{Lo93}, the quantities $H_{ij}(\PP^2)$ and $H_{j}(\PP^2)$ can be derived and expressed as
    {\small$H_{ij}(\PP^2) = d_{\rm sum}\,\sum_{k=2}^{N}\frac{1}{1-\lambda_k^2}\kh{\frac{\psi_{kj}^2}{d_j}-\frac{\psi_{ki}\psi_{kj}}{\sqrt{d_id_j}}}$}
and
    {\small$H_{i}(\G') = \frac{1}{\pi_i} \sum_{k=2}^{N} \frac{1}{1-\lambda_k^2}\psi_{ki}^2= \frac{1}{\pi_i}\calLL_{ii}^\dagger(\G')$}, respectively.  %
Therefore, the opinion disagreement $\delta$ can be represented as
\begin{equation}\label{eq:delta02}
    \delta=\sum_{i=1}^{N}\pi_i \sum_{k=2}^{N} \frac{1}{1-\lambda_k^2}\psi_{ki}^2=\sum_{i=1}^{N}\pi_i \calLL_{ii}^\dagger(\G')\,.
\end{equation}
In order to better understand the formulation in~\eqref{eq:delta02}, we compute relevant quantities in~\eqref{eq:delta02} for the five-node line graph and report the result in Table~\ref{tab:nodedisagree}. From Table~\ref{tab:nodedisagree}, we can see that for those relevant quantities, different nodes might have distinct values, as captured by intuition.  

   \begin{table}[htbp!]
        \renewcommand{\arraystretch}{1.5}
        \centering 
    \caption{Quantities in~\eqref{eq:delta02} for the five-node line graph. }
    \label{tab:nodedisagree}
 \begin{tabular}{cccccc}
		\hline
           node $i$ &$1$ &$2$ &$3$ &$4$ &$5$  \\ 
            \hline
            $\pi_i$ &$0.125$ &$0.25$ &$0.25$ &$0.25$ &$0.125$ \\ 
            $\mathcal{L}^{\dagger}_{ii}$ &$1.25$ &$1$ &$0.5$ &$1$ &$1.25$ \\
            $\pi_i\mathcal{L}^{\dagger}_{ii}$ &$0.15625$ &$0.25$ &$0.125$ &$0.25$ &$0.15625$ \\ \hline
	\end{tabular}
\end{table}

Before closing this section, it should be pointed out that although the sum of the $N$ terms on the right-hand sides of~\eqref{eq:delta0} and~\eqref{eq:delta02} are both equal to $\delta$, the corresponding $N$ terms in the two equations are generally not equal to each other. That is, $\calLL_{ii}^\dagger(\G')$ is often not equivalent to $\var{x_i(t)-\sum_{j=1}^{N}\pi_j x_j(t)}$ for $t\to \infty$. However,~\eqref{eq:delta02} correlates opinion disagreement $\delta$ with and the spectra of normalized adjacency matrix, enhancing our understanding of network structure's influence on $\delta$. This expression for $\delta$ provides new analytical tools and insights into algorithms for computing opinion disagreement.

\section{Opinion Disagreement in Networks}

In this section, we investigate opinion disagreement of the noisy discrete-time DeGroot model on some real and model power-law networks by using~\eqref{eq:delta02}.
\begin{envblue}
To demonstrate the effect of scale-free topology on opinion disagreement, we also study the opinion disagreement on a family model networks without scale-free structures.
\end{envblue}

\begin{table*}
    \centering
    \caption{Statistics and mean relative error of \textsc{SimulateMC} (SimMC), \textsc{ApproxDelta} and \textsc{SampleDelta} with various \(\epsilon\) on some real-world networks. We indicate both node number $N$ and edge number $M$ of LCC, power-law exponent $\gamma$, opinion disagreement $\delta$ computed by the exact algorithm (Exact) through~\eqref{eq:delta02}.}
    \label{tab:real_disagree}
    \begin{tabular}{@{}crrccccccccc@{}}
        \toprule
        \multirow{3}{*}{Network}
                               &
        \multirow{3}{*}{$N$}
                               &
        \multirow{3}{*}{$M$}
                               &
        \multirow{3}{*}{$\gamma$}
                               &
        \multirow{3}{*}{$\delta$}
                               &
        \multicolumn{7}{c}{Mean relative error relative to $\delta$}                                                                                                                \\
        \cmidrule(l){6-12}
                               &        &         &       &       &
        \multirow{2}{*}{SimMC}
                               &
        \multicolumn{3}{c}{\textsc{ApproxDelta}}
                               &

        \multicolumn{3}{c}{\textsc{SampleDelta}}                                                                                                                                    \\
        \cmidrule(l){7-12}
                               &        &         &       &       &       & $\epsilon=0.35$ & $\epsilon=0.3$ & $\epsilon=0.25$ & $\epsilon=0.35$ & $\epsilon=0.3$ & $\epsilon=0.25$ \\
        \midrule
        Zachary                & 34     & 78      & 2.161 & 1.287 & 0.086 & 0.027           & 0.015          & 0.009           & 0.013           & 0.005          & 0.003           \\
        bio-celegans           & 453    & 2,025   & 2.621 & 1.132 & 0.047 & 0.059           & 0.045          & 0.034           & 0.022           & 0.021          & 0.020           \\
        web-polblogs           & 643    & 2,280   & 2.021 & 1.300 & 0.019 & 0.053           & 0.031          & 0.031           & 0.007           & 0.005          & 0.005           \\
        Protein                & 1,458  & 1,948   & 2.879 & 3.081 & 0.089 & 0.017           & 0.008          & 0.006           & 0.013           & 0.009          & 0.008           \\
        soc-hamsterster        & 2,000  & 16,097  & 2.421 & 1.140 & 0.083 & 0.006           & 0.005          & 0.002           & 0.003           & 0.002          & 0.002           \\
        socfb-Amherst41        & 2,235  & 90,954  & 2.361 & 1.013 & 0.032 & 0.004           & 0.002          & 0.002           & 0.001           & 0.001          & 0.001           \\
        Human proteins (Vidal) & 2,783  & 6,007   & 2.132 & 1.515 & 0.037 & 0.009           & 0.007          & 0.005           & 0.005           & 0.005          & 0.002           \\
        bio-grid-worm          & 3,343  & 6,437   & 2.161 & 1.681 & 0.066 & 0.056           & 0.029          & 0.027           & 0.008           & 0.006          & 0.003           \\
        socfb-Mich67           & 3,745  & 81,901  & 2.521 & 1.025 & 0.050 & 0.003           & 0.002          & 0.002           & 0.001           & 0.001          & 0.001           \\
        soc-PagesTVshow        & 3,892  & 17,239  & 2.74  & 1.766 & 0.044 & 0.007           & 0.002          & 0.002           & 0.004           & 0.003          & 0.002           \\
        web-EPA                & 4,253  & 8,897   & 2.461 & 1.791 & 0.072 & 0.031           & 0.017          & 0.007           & 0.007           & 0.005          & 0.004           \\
        web-spam               & 4,767  & 37,375  & 2.381 & 1.126 & 0.002 & 0.009           & 0.006          & 0.003           & 0.002           & 0.001          & 0.001           \\
        socfb-American75       & 6,370  & 217,654 & 2.441 & 1.016 & 0.096 & 0.006           & 0.005          & 0.002           & 0.001           & 0.001          & 0.001           \\
        socfb-MIT              & 6,402  & 251,230 & 2.221 & 1.013 & 0.030 & 0.007           & 0.005          & 0.003           & 0.001           & 0.001          & 0.001           \\
        Routeviews             & 6,474  & 13,895  & 2.462 & 1.617 & 0.062 & 0.092           & 0.058          & 0.049           & 0.005           & 0.005          & 0.003           \\
        socfb-CMU              & 6,621  & 249,959 & 2.301 & 1.014 & 0.038 & 0.007           & 0.006          & 0.004           & 0.001           & 0.001          & 0.001           \\
        soc-PagesGovernment    & 7,057  & 89,429  & 2.751 & 1.082 & 0.057 & 0.008           & 0.007          & 0.005           & 0.002           & 0.001          & 0.001           \\
        socfb-Duke14           & 9,885  & 506,437 & 2.061 & 1.010 & 0.075 & 0.007           & 0.006          & 0.005           & 0.001           & 0.001          & 0.001           \\
        socfb-Bingham82        & 10,001 & 362,892 & 2.421 & 1.015 & --    & 0.005           & 0.004          & 0.003           & 0.001           & 0.001          & 0.001           \\
        ca-HepPh               & 11,204 & 117,619 & 2.081 & 1.177 & --    & 0.018           & 0.014          & 0.008           & 0.002           & 0.002          & 0.001           \\
        socfb-Stanford3        & 11,586 & 568,309 & 2.001 & 1.011 & --    & 0.006           & 0.006          & 0.005           & 0.001           & 0.000          & 0.000           \\
        socfb-Baylor93         & 12,799 & 679,815 & 2.061 & 1.010 & --    & 0.008           & 0.008          & 0.007           & 0.001           & 0.001          & 0.000           \\
        ca-AstroPh             & 17,903 & 196,972 & 2.861 & 1.096 & --    & 0.011           & 0.007          & 0.004           & 0.002           & 0.001          & 0.001           \\
        soc-Gplus              & 23,613 & 39,182  & 2.621 & 3.654 & --    & 0.231           & 0.132          & 0.050           & 0.009           & 0.006          & 0.004           \\
        CAIDA                  & 26,475 & 53,381  & 2.509 & 1.607 & --    & 0.077           & 0.076          & 0.044           & 0.007           & 0.003          & 0.002           \\
        soc-PagesArtist        & 50,515 & 819,090 & 2.261 & 1.038 & --    & 0.010           & 0.006          & 0.006           & 0.001           & 0.000          & 0.000           \\
        Brightkite             & 56,739 & 212,945 & 2.481 & 1.269 & --    & 0.020           & 0.014          & 0.009           & 0.004           & 0.003          & 0.002           \\
        \bottomrule
    \end{tabular}
\end{table*}

\subsection{Opinion Disagreement in Power-Law Real-World Networks}
We first examine opinion disagreement for real-world networks. For this purpose,
we choose a large collection of real networks, all of which are scale-free having a power-law degree distribution $P(k)\sim k^{-\gamma}$ with $2<\gamma\leqslant 3$. All studied networks are selected from two public datasets: Koblenz Network Collection~\cite{Ku13} and the Network Repository~\cite{RoAh15}. Due to its cubic-time computation complexity,~\eqref{eq:delta02} cannot handle large-scale networks. The sizes of chosen networks are relatively small, with the largest one containing around $6\times 10^4$ nodes. The statistics of these networks are listed in Table~\ref{tab:real_disagree}. 
For each network that is disconnected originally, we compute the opinion disagreement for its largest connected components (LCC).

In Table~\ref{tab:real_disagree} we report the opinion disagreement $\delta$ for the studied real scale-free networks,
which shows that for all considered networks their opinion disagreement $\delta$ is very small. Particularly,
$\delta$ does not increase with the network size $N$, but seems to be independent of $N$ and tends to small constants.



\subsection{Opinion Disagreement in Model Networks}

We continue to study opinion disagreement for \smblue{three} random model networks. These networks include two scale-free networks: Barab{\'a}si-Albert (BA) network~\cite{BaAl99} and random Apollonian networks~\cite{ZhRoFr06}.
\begin{envblue}
    Meanwhile, we also study opinion disagreement for a growing small-world network~\cite{ZhZhShGu07} without the scale-free structure.  
\end{envblue}

\textit{Barab{\'a}si-Albert Networks.}
Initially, the BA network is a small connected graph containing $m_0 \geq m$ nodes with $m \geq 1$. At every time step, a new node with $m$ links is created and connected to $m$ different existing nodes, with
the probability of being linked to an old node $i$ proportional to the degree of $i$. The growth and preferential attachment procedures are repeated until the network grows to the ideal size $N$. For a BA network with large
$N$, the average degree approaches $2m$, and the node degrees obey a distribution of power-law form $P(k) \sim k^{-3}$ for all $m$. We calculate the opinion disagreement on various BA networks with different network size $N$ ranging from $2,000$ to $50,000$ and different parameter $m=2,3,4$. We present the numerical results for opinion disagreement on all tested BA networks in Figure~\ref{FigConBaAp}(a), which shows that the opinion disagreement of BA networks does not increase with the node number $N$ but tends to an $m$-dependent constant. 

\textit{Random Apollonian Networks.}
We proceed to study the behavior of opinion disagreement $\delta$ for the random $d$-dimensional Apollonian network with $d\ge 2$, which is constructed as follows~\cite{ZhRoFr06}. Initially, the network is a $(d+2)$-clique, the complete graph with $d+2$ nodes. A $(d+1)$-clique is called \emph{active} if it was never chosen before. At each time step, the network grows by creating a new node connecting to an active $(d+1)$-clique selected randomly. The procedure of selecting active $(d+1)$-cliques and creating new nodes is repeated, until the network grows to a desirable size $N$. The Apollonian networks have a power-law degree distribution $P(k) \sim k^{-\gamma}$ with the power exponent $\gamma=2+1/(d-1)$. 

In Figure~\ref{FigConBaAp}(b), we plot the numerical results about the opinion disagreement $\delta$ for random Apollonian networks with various $d$ and $N$. Figure~\ref{FigConBaAp}(b) indicates that as networks grow, the opinion disagreement $\delta$ does not increase, but converges to small constants that rely on the parameter $d$. 

\begin{envblue}

\textit{Growing Small-World Networks.}
We finally study  opinion disagreement  in the growing small-world network model in~\cite{ZhZhShGu07}. This network evolves through sequential growth iterations following these rules: The initial configuration consists of three nodes arranged in a triangular formation on a circular structure. In each subsequent growth iteration, a new node is inserted into a randomly chosen position between existing nodes along the perimeter. The topological evolution then proceeds through two distinct mechanisms. In the connection phase, the new node establishes undirected edges to its two immediate spatial neighbors in the existing configuration. In the subsequent rewiring phase, the model executes the operation of removing the edge  between the two neighbors of the newly connected node  with probability \(p\), which effectively introduces structural randomness while maintains connectivity. It has been proved~\cite{ZhZhShGu07} that this random network exhibits small-world properties. Furthermore, when the removing probability \(p\) converges to \(0\), the resulting random graph exhibits similar structure of the Erd\"os-R\'enyi random graph~\cite{ZhZhShGu07}. On the other hand, when \(p\) converges to \(1\), the resulting random graph is transformed into a small-world network~\cite{ZhZhShGu07}. Figure~\ref{fig:er_small_world_delta} reports the results for the opinion disagreement $\delta$ of growing small-world networks with various $n$ and $p$. As demonstrated in Figure~\ref{fig:er_small_world_delta}, the opinion disagreement $\delta$ grows logarithmically with the node number $n$ for each tested \(p\), indicating that the opinion disagreement $\delta$ of networks without scale-free structures is affected by the network size.
\end{envblue}



\subsection{Result Analysis}
The above results demonstrate that for both real and model scale-free networks, their opinion disagreement $\delta$ is small constants, independent of the number of nodes. Since for the noisy DeGroot model on a graph, their node opinions do not converge, the constant scaling is the minimum scaling that opinion disagreement can achieve. As shown in~\eqref{eq:delta02}, the opinion disagreement of a graph is determined by the eigenvalues and eigenvectors of the normalized adjacency matrix of the graph, which are in turn determined by the graph structure. We argue that the small opinion disagreement observed for the considered power-law graphs is attributed to their scale-free structure, which can be accounted for from the following heuristic arguments.

In a power-law graph $\G$, there exist some large-degree nodes linked to many other nodes in the graph, leading to a small average shortest-path distance~\cite{WaSt98}, which grows at most logarithmically with the node number $N$. The small-world phenomenon facilitates fast spread of opinions on the network. Equation~\eqref{eq:delta02x} shows that opinion disagreement $\delta$ of a graph $\G$ is closely related to the partial mean hitting time $H_i(\G')$ on another graph $\G'$ associated with $\G$. It has been shown that in a scale-free small-world graph $\G$, the hitting time between two nodes are relatively small~\cite{ShZh19}. Since $\G'$ is denser than $\G$~\cite{ChChLiPeTe15b}, $H_i(\G')$ is smaller than $H_i(\G)$, although the latter is already small in a power-law graph~\cite{TeBeVo09}.

\begin{figure}[htbp]
    \centering
    \subfigure[BA networks]
    {\includegraphics[width=0.23\textwidth]{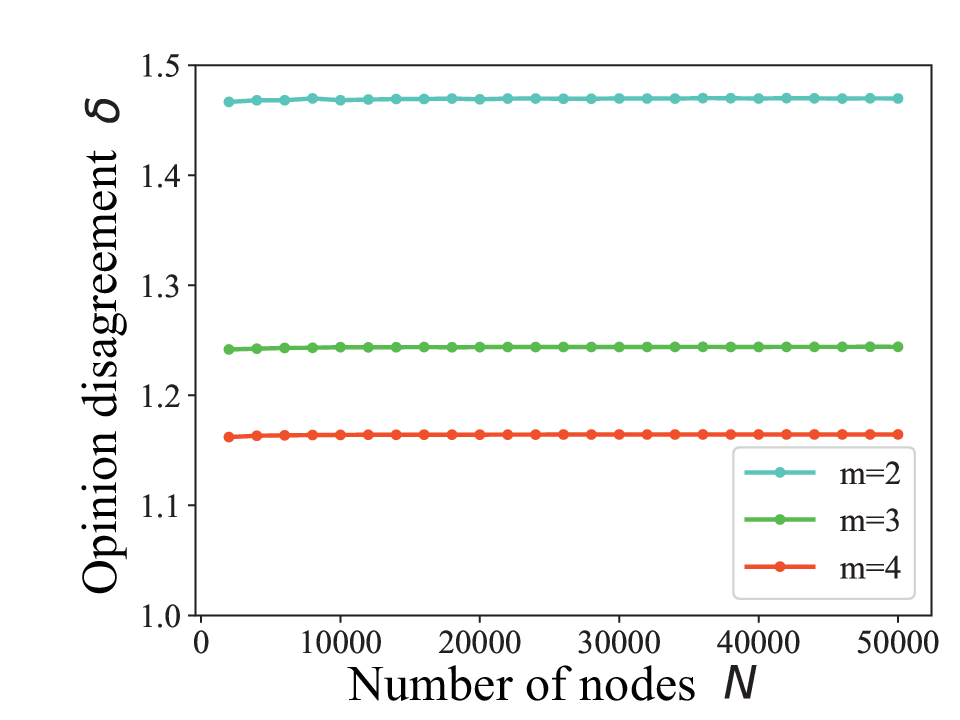} }
    \subfigure[Random Apollonian networks]
    {\includegraphics[width=0.225\textwidth]{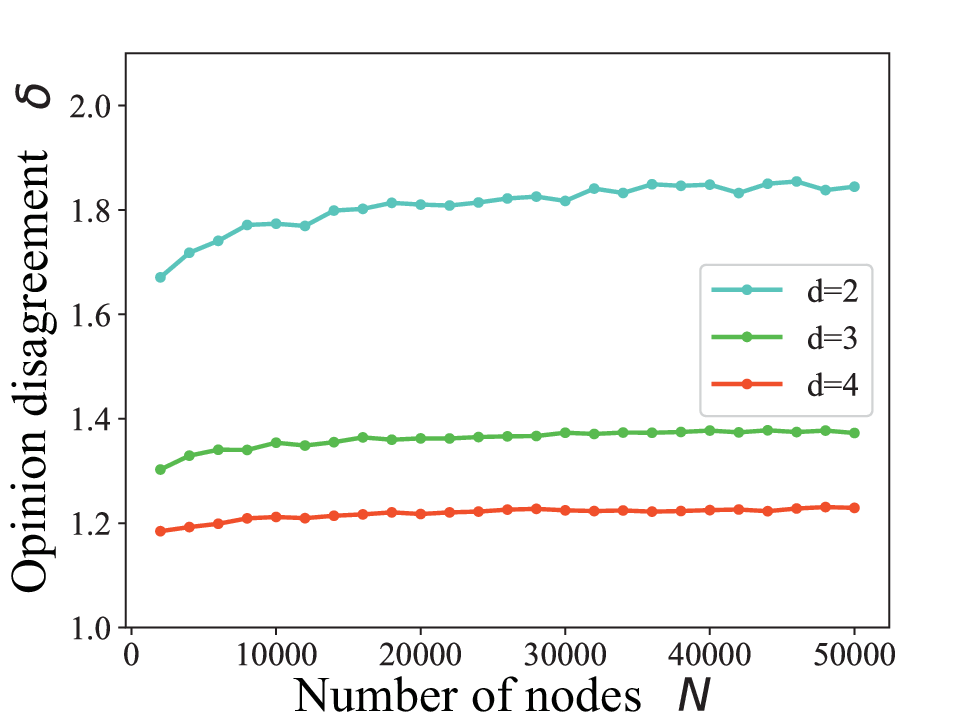}}
    \caption{Numerical results for the opinion disagreement on BA networks and random Apollonian networks.}
    \label{FigConBaAp}
\end{figure}

\begin{figure}[htbp]
    \centering
    \includegraphics[width=0.455\textwidth]{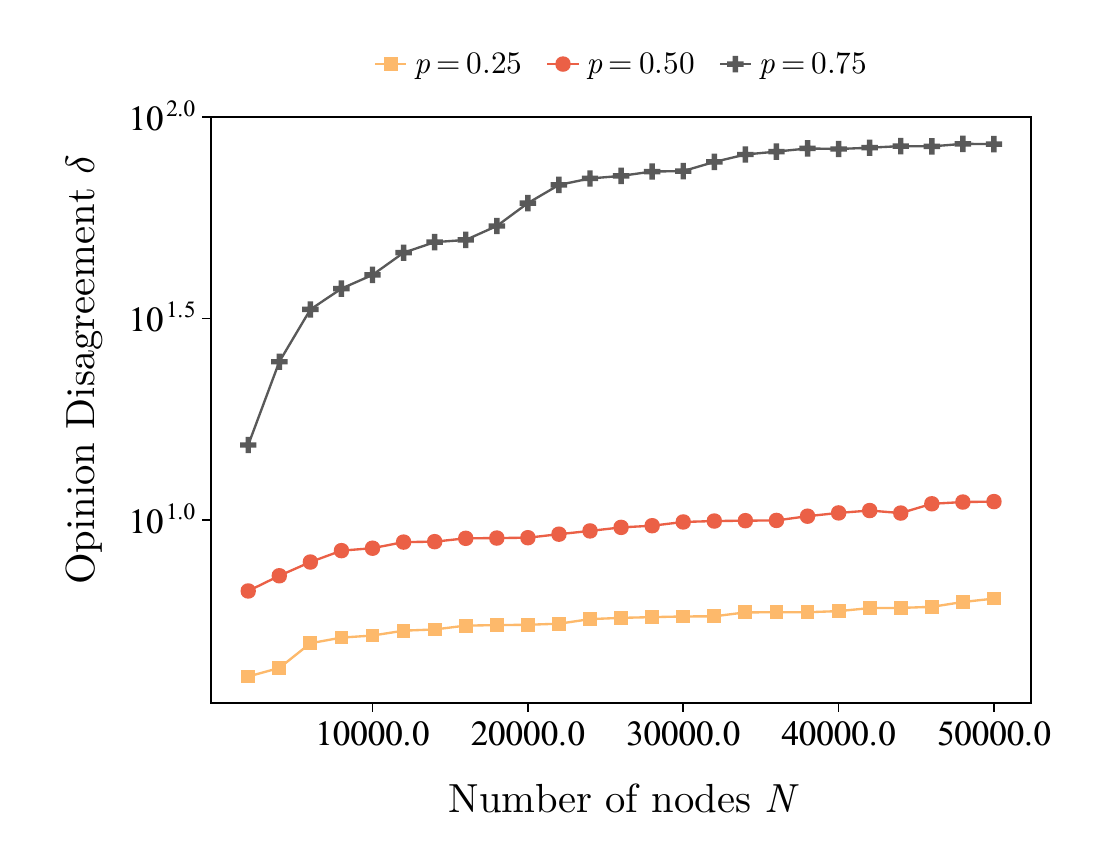}
    \caption{Numerical results for the opinion disagreement on Growing Small-World networks.}
    \label{fig:er_small_world_delta}
\end{figure}

\section{Fast Approximation Algorithms for Estimating Opinion Disagreement}

In~\eqref{eq:delta02}, we provide an explicit expression for opinion disagreement $\delta$ of a graph $\G$, in terms of the eigenvalues and eigenvectors of the normalized adjacency matrix $\SS$. However, direct computation for all eigenvalues and eigenvectors of $\SS$ takes $O(N^3)$ time and $O(N^2)$ space, which is intractable for huge networks.

To address this computation issue, in the section, we first introduce a modification of an exiting nearly linear time algorithm~\cite{ZhXuZh20} for estimating the partial mean hitting time to approximate $\delta$ of a graph $\G$.  We then present a
sublinear time random walk based sampling algorithm to estimate $\delta$.



\subsection{Modified Algorithm Based on Laplacian Solvers}

In this subsection, we introduce a nearly linear time method to estimate opinion disagreement, which serves as a baseline for our sampling algorithm proposed in next subsection.

\subsubsection{Alternative Expression for 
Relevant Quantity}
In~\eqref{eq:delta}, it is shown that the problem of computing $\delta$ can be reduced to determining the partial mean hitting time $H_i(\G')$ on a graph $\G'$ derived from the original graph $\G$, and a recent work~\cite{ZhXuZh20} presents a nearly linear time algorithm to estimate the partial mean hitting time $H_i(\G)$ for all nodes in a graph $\G$. Unfortunately, the algorithm in~\cite{ZhXuZh20} only holds for those graphs without any self-loop, which is not applicable to $H_i(\G')$, since $\G'$ contains self-loops. In order to invoke the approach in~\cite{ZhXuZh20}, a natural idea is to remove self-loops in $\G'$. However, after deleting all self-loops from graph $\G'$, many properties of $\G'$ may no longer be preserved, such as the transition matrix and the normalized Laplacian matrix, both of which are important for related quantities about random walks on $\G'$, including the partial mean hitting time $H_i(\G')$. Fortunately, for a graph with self-loops, the removal of self-loops has little influence on the Laplacian matrix. Below we express $H_i(\G')$ in terms of the diagonal entries of the pseudoinverse for the Laplacian matrix $\LL(\G')$ associated with graph $\G'$, which is helpful for the modification of the algorithm in~\cite{ZhXuZh20}.

Formula~\eqref{eq:delta02} represents the opinion disagreement $\delta$ of a graph $\G$ in terms of the diagonal entries of the pseudoinverse for the normalized Laplacian matrix $\calLL(\G')$ for graph $\G'$. Next we establish the relation between the diagonal entries of the normalized Laplacian matrix and the Laplacian matrix associated with a graph with or without self-loops. By using Theorem~1 in~\cite{Bo13}, it is straightforward to obtain the following result.

\begin{lemma}\cite{XuShZhKaZh20}
    \label{lem:pinv}
    For a connected undirected weighted graph $\G=(\V,\E)$ with or without self-loops, let $\LL^{\dagger}$ and $\calLL^{\dagger}$ be the Moore-Penrose inverse of its Laplacian matrix $\LL$ and normalized Laplacian matrix $\calLL$, respectively. Then,
    \small{$\calLL^{\dagger} =\Big (\II - \frac{1}{d_{\rm sum}}\DD^{\frac{1}{2}} \JJ \DD^{\frac{1}{2}}\Big) \DD^{\frac{1}{2}} \LL^{\dagger} \DD^{\frac{1}{2}}\Big(\II - \frac{1}{d_{\rm sum}}\DD^{\frac{1}{2}} \JJ \DD^{\frac{1}{2}}\Big )$}, where $\JJ=\one \one^{\top}$.
\end{lemma}

According to Lemma~\ref{lem:pinv}, we obtain the following formula {\small $\DD^{\frac{1}{2}}(\II - \frac{1}{d_{\rm sum}}\DD^{\frac{1}{2}} \JJ \DD^{\frac{1}{2}})\ee_i = \sqrt{d_i}(\ee_i - \ppi)$},
which leads to {\small $\ee_i^\top \calLL^\dagger \ee_i = d_i(\ee_i-\ppi)^\top\LL^\dagger(\ee_i-\ppi)$}.
Then, the opinion disagreement $\delta$ of a graph $\G$ can be expressed in terms of the diagonal elements of $\LL(\G')^{\dagger}$, or the quadratic forms of $\LL(\G')^{\dagger}$, as {\small $\delta=d_{\rm sum}\sum_{i=1}^{N}\pi_i^2 (\ee_i-\ppi)^\top\LL(\G')^\dagger(\ee_i-\ppi)$}, which is instrumental in the design of the modified algorithm.

\subsubsection{Spectral Graph Sparsification}
As can be seen from~\eqref{Lap}, direct computation of pseudoinverse $\LL^{\dagger}(\G')$ of $\LL(\G')$ needs the square of matrix $\PP$, which involves matrix multiplication taking $O(N^3)$ time. Thus, it is infeasible to construct the graph $\G'$ explicitly when the number of nodes $N$ in the original graph $\G$ is large. In order to reduce the computation complexity, we construct Laplacian matrix $\LL(\tilde{\G})$ corresponding to graph $\tilde{\G}$ and use it
to approximate matrix $\LL(\G')$. 

\begin{lemma}\label{eq:Ltil}~\cite{ChChLiPeTe15b} For a graph $\G'$ with Laplacian matrix $\LL(\G') = \DD - \DD \kh{\DD^{-1}\AA}^2$,
where $\AA$ and $\DD$ are, respectively, the adjacency matrix and degree matrix of the associated graph $\G$ with $N$ nodes and $M$ edges, one can construct a graph $\tilde{\G}$ with Laplacian matrix $\LL(\tilde{\G})$ having $O(M\eps^{-2}\log N)$ non-zero entries, in time $O(M\eps^{-2}\log^2 N)$, such that
\begin{align*}
(1-\eps)\, \xx^\top\LL(\G')\xx \le \xx^\top \LL(\tilde{\G})\xx \le (1+\eps)\, \xx^\top\LL(\G')\xx\, 
\end{align*}
holds  for any vector $\xx\in\mathbb{R}^N$.
\end{lemma}

In the sequel, we refer the technique in Lemma~\ref{eq:Ltil} to \textsc{Sparsify}$(\G,\eps)$, which takes a graph $\G$ and an approximation parameter $\eps$, and returns a Laplacian matrix $\LL(\tilde{\G})$ corresponding to a graph $\tilde{\G}$ that is much sparser than $\G'$. Lemma~\ref{eq:Ltil} shows that the opinion disagreement $\delta$ on graph $\G$ can be well preserved, if we compute it using $\LL^{\dagger}(\tilde{\G})$, instead of $\LL^{\dagger}(\G')$.

Lemma~\ref{eq:Ltil} returns a sparse graph $\tilde{\G}$ in nearly linear time with respect to the number of edges $M$, whose Laplacian is similar to that of $\G'$. Since graph $\tilde{\G}$ may contain self-loops, one cannot directly adopt the method in~\cite{ZhXuZh20} to compute the partial mean hitting time $H_i(\tilde{\G})$. In order to approximate opinion disagreement $\delta$ of a graph $\G$, we delete all self-loops in graph $\tilde{\G}$ to obtain graph $\bar{\G}$ with $\bar{M}$ edges. 
Let $\bar{\BB} \in \mathbb{R}^{\bar{M} \times N}$ be the edge-node incidence matrix of graph $\bar{\G}$. For every edge $e$ with end nodes $i$ and $j$, a direction is assigned arbitrarily. Then the entry $b_{eu}$ at row associated with edge $e$ and column associated with nodes $u$ is defined as follows: $b_{eu} = 1$ if $u$ is the tail of $e$, $b_{eu}=-1$ if $u$ is the head of $e$, and $b_{eu}=0$ otherwise. Let $\bar{\WW} \in \mathbb{R}^{\bar{M} \times \bar{M}}$ be the diagonal matrix, with the diagonal corresponding to edge $e$ being the weight of $e$ in $\bar{\G}$. Then the Laplacian matrix $\LL(\bar{\G})$ of $\bar{\G}$ can be written as $\LL(\bar{\G})=\LL(\tilde{\G})=\bar{\BB}^\top\bar{\WW}\bar{\BB}$. In this way, we reduce computing the diagonal element $\calLL_{ii}^\dagger(\G')$ of matrix $\calLL (\G')^{\dagger}$ in~\eqref{eq:delta02} to approximating the quadratic form $(\ee_i-\ppi)^\top\LL(\bar{\G})^\dagger(\ee_i-\ppi)$ for matrix $\LL(\bar{\G})^\dagger $ associated with graph $\bar{\G}$.

Note that in order to evaluate the opinion disagreement $\delta$ for graph \(\G\), we introduce three more graphs \(\G'\), \(\tilde{\G}\), and \(\bar{\G}\). Although in Table~\ref{tab:notation} we provide a brief introduction to these graphs, we further give the main differences and connections among these graphs. First, by construction~\cite{ChChLiPeTe15b}, every node in \(\G\) and \(\G'\) has identical degree but different number of edges, with \(\G\) being sparser than \(\G'\). Thus, the properties of \(\G\) and \(\G'\) are often quite disparate. For example, even for unweighted \(\G\), the corresponding \(\G'\) is weighted, indicating that \(\G\) and \(\G'\) have distinct degree distribution, pairwise shortest paths, and weight distribution. Second, \(\tilde{\G}\) is a spectral sparsifier of \(\G'\). In other words, \(\tilde{\G}\) is a sparse subgraph of \(\G'\), but approximates the spectral properties for  Laplacian matrix of \(\G'\). In general, \(\tilde{\G}\) may not preserve the structural properties of \(\G'\). Finally, since graph \(\bar{\G}\) is obtained from \(\tilde{\G}\) by deleting all self-loops in \(\tilde{\G}\), these two graphs have the same Laplacian matrix and almost identical structural properties.

\subsubsection{Approximation of Key Quantity}
For simplicity, we write $\LLbar=\LL(\bar{\G})$ and define $C(i)=(\ee_i-\ppi)^\top\LL(\tilde{\G})^\dagger(\ee_i-\ppi)=(\ee_i-\ppi)^\top\LLbar^\dagger(\ee_i-\ppi)$. Then, $\LLbar^{\dagger}=\left(\LLbar+\frac{1}{N} \JJ\right)^{-1}-\frac{1}{N} \JJ$, $\LLbar \LLbar^{\dagger}=\LLbar^{\dagger} \LL=\II-\frac{1}{N} \JJ$, and
${\LLbar} \JJ = \JJ {\LLbar}={\LLbar}^{\dagger} \JJ = \JJ {\LLbar}^{\dagger}= \mathbf{O}$, where $\mathbf{O}$ denotes the zero matrix~\cite{GhBoSa08}. 
Thus, for each node $i$ in graph $\tilde{\G}$ or $\bar{\G}$, the quantity $C(i)$ can be represented as {\small $ C(i)=(\ee_i-\ppi)^\top\LLbar ^\dagger(\ee_i-\ppi)=(\ee_i - \ppi)^{\top} \LLbar^{\dagger} \LLbar \LLbar^{\dagger} (\ee_i - \ppi)
    =(\ee_i - \ppi)^{\top} \LLbar^{\dagger}\bar{ \BB}^{\top}\bar{ \WW} \bar{\BB} \LLbar^{\dagger} (\ee_i - \ppi)
    =\|\bar{\WW}^{\frac{1}{2}} \bar{\BB} \bar{\LL}^{\dagger} (\ee_i - \ppi)\|^2$},
which expresses the quantity $C(i)$ in terms of the Euclidian distance between two vectors $\bar{\WW}^{\frac{1}{2}} \bar{\BB} \bar{\LL}^{\dagger} \ee_i$ and $\bar{\WW}^{\frac{1}{2}} \bar{\BB} \bar{\LL}^{\dagger} \ppi$ of dimension $\bar{M}$.

By using the Johnson-Lindenstraus (JL) Lemma~\cite{JoLi84}, 
the Euclidian distance between two vectors $\bar{\WW}^{\frac{1}{2}} \bar{\BB} \bar{\LL}^{\dagger} \ee_i$ and $\bar{\WW}^{\frac{1}{2}} \bar{\BB} \bar{\LL}^{\dagger} \ppi$ can be well approximated by the distance between two vectors of low dimension. Let $\QQ$ be a
$k\times \bar{M}$ matrix with $k\ge 24\log N/\epsilon^2$, where each entry is equal to $1/\sqrt{k}$ or $- 1/\sqrt{k}$ with identical probability $1/2$. Then,
{\small $
    (1-\epsilon) C(i)
    \leq
    \|\QQ\bar{\WW}^{\frac{1}{2}} \bar{\BB} \bar{\LL}^{\dagger}(\ee_{i} - \ppi)\|^{2}
    \leq
    (1+\epsilon)  C(i)
$} holds with probability at least $1-1/N$.

Despite the fact that $\|\QQ\bar{\WW}^{\frac{1}{2}} \bar{\BB} \bar{\LL}^{\dagger}(\ee_{i} - \ppi)\|^{2}$ is a good approximation for $\|\bar{\WW}^{\frac{1}{2}} \bar{\BB} \bar{\LL}^{\dagger}(\ee_{i} - \ppi)\|^{2}$ and $\QQ\bar{\WW}^{\frac{1}{2}} \bar{\BB}$ can be quickly computed by sparse matrix multiplication, directly computing matrix $\ZZ=\QQ\bar{\WW}^{\frac{1}{2}} \bar{\BB} \bar{\LL}^{\dagger}$ involves inverting $\bar{\LL}+\frac{1}{N}\JJ$ to obtain $\bar{\LL}^{\dagger}$. In order to avoid matrix inverse operation, by using Laplacian solvers~\cite{KoMiPe11,CoKyMiPaPeRaSu14},  
we alternatively solve the system of equations $\bar{\LL} \zz_i=\qq_i$, $i=1,2\ldots,k$, where $\zz^\top_i$ and $\qq^\top_i$ are, respectively, the $i$th row of $\ZZ$ and $\QQ\bar{\WW}^{\frac{1}{2}} \bar{\BB}$. Here we use the Laplacian solver in~\cite{CoKyMiPaPeRaSu14}, an algorithm $\xx = \mathtt{LaplSolve}(\bar{\LL},\yy,\kappa)$ with complexity $\Otil \left(M \log(1/\kappa) \right)$, where the notation $\Otil(\cdot)$ hides $\mathrm{poly} \log $ factors. It takes a Laplacian matrix $\bar{\LL}$, a column vector $\yy$, and an error parameter $\kappa > 0$ as inputs, and returns a column vector $\xx$ satisfying $\boldsymbol{1}^\top\xx = 0$ and
$\|\xx - \bar{\LL}^{\dagger} \yy\|_{\bar{\LL}} \leq \kappa \|\bar{\LL}^{\dagger} \yy\|_{\bar{\LL}}$,
where $\norm{\yy}_{\bar{\LL}} = \sqrt{\yy^{\top} \bar{\LL} \yy}$.

For graph $\bar{\G}$, let $\bar{d}_i$ be the degree of node $i$, and $\bar{d}_{\rm sum}$ the total degree of all nodes, and let $w_{\rm max}$ and $w_{\rm min}$, denote the maximum and minimum edge weight, respectively.
Combining the JL lemma and the Laplacian solver, $C(i)$ can be efficiently approximated, as stated in the following lemma. 

\begin{lemma}\cite{ZhXuZh20}\label{lem:error1}
Given an approximate factor $\epsilon \le 1/2$ and a $k\times N$ matrix $\ZZ$ that satisfies
$(1-\epsilon) C(i)\leq \|\ZZ (\ee_{i} - \ppi)\|^{2} \leq (1+\epsilon) C(i)$
    for any node $i\in \V$ of graph $\bar{\G}$, as well as
    {\small
    $\quad (1-\epsilon) \|\bar{\WW}^{\frac{1}{2}}\bar{ \BB} \bar{\LL}^{\dagger} (\ee_i-\ee_j)\|^2
        \leq
        \|\ZZ (\ee_i - \ee_j)\|^2
        \leq
        (1+\epsilon) \|\bar{\WW}^{\frac{1}{2}} \bar{\BB}\bar{ \LL}^{\dagger} (\ee_i-\ee_j)\|^2$
    } 
    for any pair of nodes $i,j \in \V$.
    Let $\zz_i$ be the $i$th row of $\ZZ$ and let $\tilde{\zz}_i$ be an approximation of $\zz_i$ for all $i \in \{1,2,...,k\}$, satisfying $ \|\zz_i-\tilde{\zz}_i\|_{\bar{\LL}}\le \kappa
        \|\zz_{i}\|_{\bar{\LL}}$,
    where
    \begin{equation}\label{EE14}
        \kappa \leq \frac{\epsilon }{3} \frac{\bar{d}_{\rm sum}-\bar{d}_i}{\bar{d}_{\rm sum}}
        \sqrt{\frac{(1-\epsilon) w_{\min}}{(1+\epsilon) N^4 w_{\max}}}.
    \end{equation}
    Then for any node $i \in \V$,
    {\small $
                (1 - \epsilon)^2 C(i)
                \leq
                \|\ZZtil (\ee_{i} - \ppi)\|^2
                \leq
                (1 + \epsilon)^2 C(i)
            $},
    where $\ZZtil = [\tilde{\zz}_1, \tilde{\zz}_2, ..., \tilde{\zz}_k]^\top$.
\end{lemma}

\subsubsection{Modified Algorithm}
Based on the aforementioned techniques, we introduce an algorithm \textsc{ApproxDelta}$(\calG, \epsilon)$ to estimate the opinion disagreement $\delta$ for an arbitrary graph $\G$. The pseudocode of the approximation algorithm is presented in Algorithm~1 and the performance is characterized in Theorem~\ref{th:alg}.

\begin{theorem}\label{th:alg}
    There is an algorithm \textsc{ApproxDelta}$(\G,\eps)$ with $\Otil(M\eps^{-2}\log c)$ time and  $O(M)$ space complexity, which inputs a scalar $0<\eps<1$ and a graph $\G=(\V,\E,w)$ where $c=\frac{w_{\max}}{w_{\min}}$ and returns an approximation $\tilde{\delta}$ for opinion disagreement $\delta$ of the graph $\G$, such that with high probability,
    {\small $(1-\eps)^3\delta\le\tilde{\delta}\le(1+\eps)^3\delta $}
    for the graph $\G$. 
\end{theorem}

\begin{algorithm}
    \caption{\textsc{ApproxDelta}$(\calG, \epsilon)$}
    \label{alg:approx_delta}
    \Input{
        $\G(\V, \E, w) $: a connected undirected weighted graph \\
        $\epsilon$: an approximation parameter \\
    }
    \Output{
        $\tilde{\delta}$: approximation of opinion disagreement $\delta$ in $\G$
    }
    $\tilde{\G}=$ \textsc{Sparisify}$(\G,\eps)$\;
    $\bar{\G}$: graph obtained from $\tilde{\G}$ by removing self-loops \;
    $\bar{\LL}=$ Laplacian of $\bar{\G}$\;
    Construct a matrix $\QQ_{k \times \bar{M}}$, where $k=\lceil 24\log N/\epsilon^2 \rceil$ and each entry is $\pm 1/\sqrt{k}$ with identical probability\;
    \For{$i=1$ to $k$}{
        $\qq_i^\top $=the $i$-th row of $\QQ_{k \times \bar{M}} \bar{\WW}^{1/2} \bar{\BB}$ \\
        $\tilde{\zz}_i=\mathtt{LaplSolve}(\bar{\LL}, \qq_i, \kappa)$ where parameter $\kappa$ is given by~\eqref{EE14}\\
    }
    Calculate the constant vector ${\vp}=\ZZtil \ppi$\;
    \ForEach{$i \in \V$}{
    $\tilde{C}(i)= \|\ZZtil_{:,u}-{\vp}\|^2$
    }
    $\tilde{\delta}=d_{\rm sum}\,\sum_{i=1}^N \pi_i^2\cdot \tilde{C}(i)$\;
    \Return $\tilde{\delta}$
\end{algorithm}

\subsection{Algorithm Based on Sampling Random walks}
Although $\textsc{ApproxDelta}$ significantly reduces the running time of the accurate algorithm by directly computing the all eigenvalues and eigenvector according to~\eqref{eq:delta02}, we observe in our experiments (Section 7) that $\textsc{ApproxDelta}$ is still intolerable for large networks with ten million nodes.
We next propose an efficient and effective random-walk based sampling algorithm estimating opinion disagreement.

\subsubsection{Reformulation and Approximation for Key Quantity}
We first present an alternative expression for matrix $\calLL ^\dagger(\G')$.


\begin{lemma}\label{lem:NorL}
    The pseudoinverse $\calLL^\dagger(\G')$ of matrix $\calLL(\G')$ can be expressed as
        {\small $\calLL^\dagger(\G') = \sum_{i=0}^\infty\left(\mS^{2i}-\ppsi_1\ppsi_1^\top\right)$}.
\end{lemma}
\proof
For any positive integer $i$, matrix $\mS^{2i}$ can be written as $\mS^{2i}=\sum_{j=1}^N\lambda_j^{2i}\ppsi_j\ppsi_j^\top=\ppsi_1\ppsi_1^\top+\sum_{j=2}^{N}\lambda_j^{2i}\ppsi_j\ppsi_j^\top$. Then, one has
\begin{equation*}
    \begin{aligned}
        &\calLL^\dagger(\G')
        =  (\mI-\mS^2)^\dagger
        =\sum\nolimits_{j=2}^N\frac{1}{1-\lambda_j^2}\ppsi_j\ppsi_j^\top            \\
        = & \sum\nolimits_{j=2}^N\sum\nolimits_{i=0}^\infty\lambda^{2i}_j\ppsi_j\ppsi_j^\top
        =\sum\nolimits_{i=0}^\infty\sum\nolimits_{j=2}^N\lambda^{2i}_j\ppsi_j\ppsi_j^\top    \\
        = & \sum\nolimits_{i=0}^\infty\left(\mS^{2i}-\ppsi_1\ppsi_1^\top\right),
    \end{aligned}
\end{equation*}
which completes the proof.
\endproof

Since for any positive integer $i$, $\mS^{2i} = \mD^{1/2}\mP^{2i}\mD^{-1/2}$, the corresponding diagonal entries of matrices $\mS^{2i}$ and $\mP^{2i}$ are equal to each other. That is, $\mS^{2i}_{kk}=\mP^{2i}_{kk}$ for any positive integer $k$. By definition of matrix $\mP$, $\mP^{2i}_{kk}$ is the probability of a random walk starting from node $k$ and returning to $k$ at the $2i$-th step. Then, by Lemma~\ref{lem:NorL}, the $k$th diagonal entry of $\calLL^\dagger(\G')$ is
$\calLL^\dagger_{kk}(\G')= \sum_{j=0}^{\ell-1} (\mP_{kk}^{2j}-\ppi_k)+\sum_{j= \ell}^{\infty} (\mP_{kk}^{2j}-\ppi_k)$, for any $\ell >0$. Below we show that for an additive error $\epsilon>0$, we can choose $\ell$ appropriately such that any diagonal entry of $\calLL^\dagger(\G')$ can be approximated by the first sum within additive error $\epsilon/2$.

\begin{lemma}\label{lem:LdiagP}
    For any $\epsilon>0$, if $\ell\geq\frac{\log{(2/(\epsilon-\epsilon\lambda))}}{2\log{(1/\lambda)}}$, then for any $i\in \{1,2,\ldots,N\}$,
    {\small $\abs{\calLL^\dagger_{ii}(\G')-\sum_{j=0}^{\ell-1}\left(\mP_{ii}^{2j}-\ppi_i\right)}\leq\frac{\epsilon}{2}$}.
\end{lemma}

\proof
From Lemma~\ref{lem:NorL}, for any $i\in \{1,2,\ldots,N\}$, we have
\begin{equation*}
    \begin{aligned}
          & \abs{\calLL^\dagger_{ii}(\G')-\sum_{j=0}^{\ell-1}\left(\mP_{ii}^{2j}-\ppi_i\right)}                                                           \\
        = & \abs{\sum_{j=0}^\infty\left(\mS^{2j}_{ii}-\ppi_i\right)-\sum_{j=0}^{\ell-1}\left(\left(\mD^{-1/2}\mP^{2j}\mD^{1/2}\right)_{ii}-\ppi_i\right)} \\
        = & \sum_{j=\ell}^\infty\left(\mS^{2j}_{ii}-\ppi_i\right)=\ve_i^\top\sum_{j=\ell}^\infty\left(\mS^{2j}-\ppi_i\right)\ve_i.
    \end{aligned}
\end{equation*}
Write $\ve_i $ in terms of $\ppsi_1$, $\ppsi_2$, $\cdots$, $\ppsi_N$ as $\ve_i = \sum_{k=1}^N\alpha_k\ppsi_k$. It is easy to derive that $\alpha_1=\ppsi_1^\top\ve_i=\ppi_i$. Then $\ve_i^\top(\mS^{2j}-\ppi_i)\ve_i=\sum_{k=2}^N\alpha_k^2\lambda_k^{2j}$. Moreover, we have $\sum_{k=2}^N\alpha_k^2\leq\sum_{k=1}^N\alpha_k^2=\ve_i^\top\ve_i=1$. Thus,
\begin{equation*}
    \begin{aligned}
             & \ve_i^\top\sum\nolimits_{j=\ell}^\infty\left(\mS^{2j}-\ppi_i\right)\ve_i
        =\sum\nolimits_{j=\ell}^{\infty}\sum\nolimits_{k=2}^N\alpha_k^2\lambda_k^{2j}            \\
        \leq & \sum\nolimits_{j=\ell}^\infty\lambda^{2j}\sum\nolimits_{k=2}^N\alpha_k^2
        \leq \frac{\lambda^{2\ell}}{1-\lambda}\leq\frac{\epsilon}{2},
    \end{aligned}
\end{equation*}
where the last inequality follows from $\ell\geq\frac{\log{(2/(\epsilon-\epsilon\lambda))}}{2\log{(1/\lambda)}}.$
\endproof

\subsubsection{Random walks Sampling Algorithm}
Lemma~\ref{lem:LdiagP} shows that $\sum_{j=0}^{\ell-1}(\mP_{ii}^{2j}-\ppi_i)$ is a good approximation for the $i^{\mathrm{th}}$ diagonal element $\calLL^\dagger_{ii}(\G')$ of $\calLL^\dagger(\G')$. Since $\mP_{ii}^{2j}$ can be estimated by sampling random walks, we present a sampling algorithm for approximating $\delta$ by estimating each diagonal entry and then computing their weighted sum $\overline{\delta}=\sum_{i=1}^{N}\ppi_i\sum_{j=0}^{\ell-1}\left(\mP_{ii}^{2j}-\ppi_i\right)$. As shown in the following lemma, $\overline{\delta}$ has an error bound.

\begin{lemma}\label{delta}
    For any $\epsilon>0$ and $\ell\geq\frac{\log{(2/(\epsilon-\epsilon\lambda))}}{2\log{(1/\lambda)}}$, the relation $\abs{\delta-\overline{\delta}}\leq \epsilon/2$ holds.
\end{lemma}

\proof
From Lemma~\ref{lem:LdiagP}, we have
\begin{equation*}
    \begin{aligned}
        \left|\delta-\hat{\delta}\right|  = & \sum\nolimits_{i=1}^{N}\ppi_i\left|\calLL^\dagger_{ii}(\G')-\sum\nolimits_{j=0}^{\ell-1}\left(\mP_{ii}^{2j}-\ppi_i\right)\right| \\
        \leq                         & \frac{\epsilon}{2}\sum\nolimits_{i=1}^N\ppi_i=\frac{\epsilon}{2},
    \end{aligned}
\end{equation*}
which completes the proof.
\endproof

However, if we approximate the opinion disagreement by directly estimating $\sum_{i=1}^N\ppi_i\sum_{j=0}^{\ell-1}\left(\mP^{2j}_{ii}-\ppi_i\right)$, we need to sample
$\ell$-truncated random walks for each of the $N$ nodes, which is unacceptable for large networks. In
fact, we can evaluate the opinion disagreement by only summing a small
number of individuals. To achieve this goal, we introduce the
following lemmas.

\begin{lemma}(Hoeffding's inequality~\cite{Ho63})
    \label{lem:Hoeffding}
    Let $x_1, x_2, \ldots , x_r$ be $r$ independently random variables, and let $x = \sum_{i=1}^rx_i$. If
    $a_i \leq x_i \leq b_i$ holds for every $1 \leq i \leq r$, then for any $\alpha > 0$,
    \begin{equation*}
        \bbP\left(\abs{x-\bbE(x)}\geq\alpha\right)\leq 2e^{-\frac{2\alpha^2}{\sum_{i=1}^r\Delta_i^2}},
    \end{equation*}
    where $\Delta_i=b_i-a_i$.
\end{lemma}

\begin{lemma}\label{lem:sqrtn}
    Given $N$ positive numbers $x_1, x_2, \ldots , x_N$ such that $x_i\in [0,\Delta_i]$ with
    their sum $x = 	\sum_{i=1}^Nx_i$, an error $\beta > N^{-1/2}\log^{1/2}{N}$, let $\Delta=(\sum_{i=1}^N\Delta_i^2)^{1/2}$. If we select $t =
        O(\Delta N^{1/2}\log^{1/2}{N}/\beta)$ numbers, $x_{c_1}, x_{c_2},\ldots,x_{c_t}$ by Bernoulli trails with
    success probability $p = \Delta N^{-1/2}\log^{1/2}{(2N)}/\beta$ satisfying $0 < p < 1$, and define $\tilde{x} = \sum_{i=1}^tx_{c_i}/  (\Delta N^{-1/2}\log^{1/2}{(2N)}/\beta)$, then $\tilde{x}$ is an approximation of the sum $x$ of the original $N$ numbers, satisfying $\abs{x - \tilde{x}} \leq N^{1/2}\beta$.
\end{lemma}

\proof
For each $i\in\{1,2,\ldots,N\}$, let $y_i$ be $N$ Bernoulli random
variables such that $\bbP(y_i = 1) = p$ and $\bbP(y_i = 0) = 1-p$, where
$y_i = 1$ indicates that the $x_i$ is selected, and $y_i = 0$ otherwise.
Let $z_i = x_iy_i$ be $N$ independent random variables satisfying
$z_i\in [0, x_i]\subseteq[0, \Delta_i]$. Denote $y$ as the sum of the $N$ random variables
$y_i$, and denote $a$ as the sum of the $N$ random variables $z_i$. Namely, $y=\sum_{i=1}^Ny_i$ and $z=\sum_{i=1}^Nz_i$. By definition, $y$ and $z$ represent, respectively, the
number of selected numbers and their sum. Then the expectation
of $y$ is $\bbE(y) = Np$, and the expectation of $z$ is $\bbE(z) = px$. According
to Hoeffding's inequality in Lemma~\ref{lem:Hoeffding}, we obtain
\begin{equation*}
    \begin{aligned}
        \bbP(\abs{\tilde{x} - x} \geq N^{1/2}\beta)
        = & \bbP(\abs{z - px}/p \geq N^{1/2}\beta) \\ \leq& 2e^{-\frac{2p^2N\beta^2}{\sum_i\Delta_i^2}} \leq
        \frac{1}{N},
    \end{aligned}
\end{equation*}
finishing the proof.
\endproof

\begin{algorithm}
    \caption{\textsc{SampleDelta}$(\calG, \epsilon)$}
    \label{alg:sample_delta}
    \Input{
        $\G(\V, \E, w) $: a connected undirected weighted graph \\
        $\epsilon$: an approximation parameter \\
    }
    \Output{
        $\hat{\delta}$: approximation of opinion disagreement $\delta$ in $\G$
    }
    $\ell\gets\frac{\log{(2/(\epsilon-\epsilon\lambda))}}{2\log{(1/\lambda)}}$, $r\gets\frac{2\ell^2\log{(2N^2\ell)}}{\epsilon^2}$\;
    Sample a node set $\cX\subset\cV$ satisfying $\abs{\cX} = \lceil \frac{N^{1/2}\log^{1/2}{N}}{(1-\lambda)\epsilon}\rceil$ \;
    \For {each node $i\in\cX$}{
        \For{$j=1$ to $\ell$}{
            \For{$k=1$ to $r$}{
                $u=i$\;
                \For{$m=1$ to $2j$}{
                    \If{$u=i$}{
                        $\hat{\delta}^{(j)}_i\gets \hat{\delta}^{(j)}_i+1 $\;
                    }
                    $u\gets$ a randomly selected neighbor of $u$\;              
                }
            }
        }
    }
    \Return $\hat{\delta} = \frac{N}{\abs{\cX}}\sum_{i\in\cX}\ppi_i\sum_{j=0}^{\ell-1}\left(\hat{\delta}^{(j)}_i/r-\ppi_i\right)$
\end{algorithm}

Lemma~\ref{lem:sqrtn} reveals that computing $\overline{\delta}$ does not require estimating all $N$ diagonal elements of $\calLL^\dagger(\G')$. Instead, only a small subset of randomly selected elements of $\calLL^\dagger(\G')$ needs to be estimated. Based on this result, we present a fast sampling algorithm \textsc{SampleDelta} to approximate opinion disagreement $\delta$, the pseudocode of which is provided in Algorithm~2. The performance of Algorithm~2 is summarized in Theorem~\ref{thm:alg2error}.

\begin{theorem}\label{thm:alg2error}
    For any error parameter $\epsilon >0$, chose $\ell \geq \frac{\log{(2/(\epsilon-\epsilon\lambda))}}{2\log{(1/\lambda)}}$. If $r\geq \frac{2\ell^2\log{(2N)}}{\epsilon^2}$, then Algorithm~2 outputs $\hat{\delta}$ as an approximation of opinion disagreement $\delta$ satisfying $|\delta-\hat{\delta}| \leq(\sqrt{N}+1)\epsilon$ with probability $\left(1-\frac{1}{N}\right)^2$. And the running time of Algorithm~2 is $O(\ell^4\sqrt{N}\log^{3/2}{N}/\epsilon^2)$.
\end{theorem}
\begin{envblue}
    \begin{proof}
We outline the proof in five phases: truncation error control, random walk estimation, node sampling, error composition, and runtime analysis.

        \textit{Truncation error control.}
        The exact opinion disagreement $\delta$ involves an infinite series $\sum_{j=0}^\infty (\mP_{ii}^{2j} - \pi_i)$. By Lemma~\ref{lem:LdiagP}, truncating this series at $\ell$ terms introduces an additive error $\leq \epsilon/2$ when $\ell \geq \frac{\log(2/(\epsilon(1-\lambda)))}{2\log(1/\lambda)}$. This follows from the geometric decay of the series, governed by the spectral gap $1-\lambda$. 

        \textit{Random walk estimation.}
        For each node $i$ and walk length $2j \leq 2\ell$, Algorithm~2 estimates $\mP_{ii}^{2j}$ by performing $r$ random walks. By Hoeffding's inequality (Lemma~\ref{lem:Hoeffding}), setting $r \geq \frac{2\ell^2\log(2N)}{\epsilon^2}$ ensures that the deviation between the estimated and true return probabilities satisfies $|\hat{\delta}_i^{(j)}/r - \mP_{ii}^{2j}| \leq \epsilon/(2\ell)$ with probability $\geq 1-1/N$. Union bounding over all $N\ell$ terms guarantees simultaneous accuracy for all estimation with probability $\geq 1-1/N$.

        \textit{Node sampling.}
        Instead of processing all $N$ nodes, Algorithm~2 uniformly samples $O(\sqrt{N}\log^{1/2}N)$ nodes. By Lemma~\ref{lem:sqrtn}, this sampling approximates the weighted sum $\sum_{i=1}^N \pi_i(\cdot)$ with an error less than $\sqrt{N}\epsilon$, as the variance of node contributions is bounded by $\Delta = 1/(1-\lambda)$. This step exploits the scale-free property where high-degree nodes dominate the weighted sum.

        \textit{Error composition.}
        The total error combines truncation and sampling errors:
        \begin{equation*}
            |\delta - \hat{\delta}| \leq |\delta - \overline{\delta}| + |\overline{\delta} - \hat{\delta}| \leq (\sqrt{N}+1)\epsilon.
        \end{equation*}
        Both error terms hold simultaneously with probability $\geq (1-1/N)^2$ by union bound.

        \textit{Runtime analysis.}
        For each of $O(\sqrt{N}\log^{1/2}N)$ sampled nodes, the algorithm performs $r = O(\ell^2\log N/\epsilon^2)$ random walks with maximum length $2\ell$. The total time complexity becomes:
        \begin{equation*}
            O\left(\sqrt{N}\log^{1/2}N \cdot \frac{\ell^2\log N}{\epsilon^2} \cdot \ell^2\right) = O\left(\frac{\ell^4\sqrt{N}\log^{3/2}N}{\epsilon^2}\right).
        \end{equation*}

        This completes the proof that Algorithm~2 achieves sublinear runtime while preserving the error guarantee.
    \end{proof}
\end{envblue}

 Theorem~\ref{thm:alg2error} indicates that Algorithm~2 approximates opinion disagreement $\delta$ with bounded error $(\sqrt{N}+1)\epsilon$ and sublinear runtime.

\section{Experimental Results}

In this section, we evaluate the performance of our proposed algorithms in terms of the approximation solution quality and running time. To this end, we compare our algorithms \textsc{SampleDelta} and \textsc{ApproxDelta} against the exact algorithm given in~\eqref{eq:delta02} as well as other baselines, by performing extensive experiments on real and model networks.

\subsection{Experimental Settings}

\begin{envblue}
    \textit{Datasets.}
    Our analysis employs graph datasets sourced from two established repositories: the Koblenz Network Collection~\cite{Ku13} and Network Repository~\cite{RoAh15}. For networks exhibiting multiple disconnected components, we perform our experiments on their largest connected components (LCCs). Key statistics of these processed networks appear in Tables~\ref{tab:real_disagree} and~\ref{tab:large_disagree}, organized by ascending network size. 
\end{envblue}

\begin{envblue}
    \textit{Environment.}
    We executed all numerical tests on a Linux server equipped with a 72-core 2.1GHz CPU and 256GB of RAM. Our implementation leverages the Julia programming environment for efficient matrix computations, particularly benefiting from the Laplacians.jl package's optimized solvers~\cite{KySa16}. The source code is publicly accessible at \url{https://github.com/vivian1tsui/opinion_disgreement}.
\end{envblue}

\begin{envblue}
   \textit{Baselines and Parameters.}
    To establish performance benchmarks, we evaluate our algorithms against conventional approximation techniques. Drawing from the theoretical formulation in~\eqref{eq:delta}, which expresses opinion disagreement through node-pair hitting times, we implement \textsc{SimulateMC}. \textsc{SimulateMC} estimates opinion disagreement by approximating pairwise node interactions through repeated random walk simulations.
\end{envblue}
For our algorithm \textsc{SampleDelta}, we also need to determine the truncated length $\ell$, which depends on $\lambda$ and $\epsilon$. In our experiments, we set $\lambda=0.9998$, which is very conservative since it is significantly larger than the actual values for most real networks~\cite{MoYuKi10}, including those with more than twenty million nodes.

\begin{table*}
    \centering
    \tabcolsep=2pt
    \fontsize{7.5}{8.3}\selectfont
    \caption{The running time (seconds) and estimated opinion disagreement of \textsc{ApproxDelta} and \textsc{SampleDelta} with various \(\epsilon\) on large connected real networks. For each network, we indicate the number of nodes as $N$ and the number of edges as $M$.}
    \label{tab:large_disagree}
    \begin{tabular}{@{}crrcccccccccccc@{}}
        \toprule
        \multirow{3}{*}{Network}
                            &
        \multirow{3}{*}{$N$}
                            &
        \multirow{3}{*}{$M$}
                            &
        \multicolumn{6}{c}{Running time (seconds)}
                            &
        \multicolumn{6}{c}{Estimated opinion disagreement}
        \\
        \cmidrule(l){4-15}
                            &
                            &
                            &
        \multicolumn{3}{c}{\textsc{ApproxDelta}}
                            &
        \multicolumn{3}{c}{\textsc{SampleDelta}}
                            &
        \multicolumn{3}{c}{\textsc{ApproxDelta}}
                            &
        \multicolumn{3}{c}{\textsc{SampleDelta}}
        \\
        \cmidrule(l){4-15}
                            &
                            &
                            &
        $\epsilon=0.35$
                            &
        $\epsilon=0.3$
                            &
        $\epsilon=0.25$
                            &
        $\epsilon=0.35$
                            &
        $\epsilon=0.3$
                            &
        $\epsilon=0.25$
                            &
        $\epsilon=0.35$
                            &
        $\epsilon=0.3$
                            &
        $\epsilon=0.25$
                            &
        $\epsilon=0.35$
                            &
        $\epsilon=0.3$
                            &
        $\epsilon=0.25$
        \\
        \midrule
        soc-twitter-follows & 404,719    & 713,319    & 261.9 & 362.8 & 605.2 & {60.21} &{99.40} & {182.5} & 2.144 & 2.304 & 2.388 & 2.387 & 2.388 & 2.383 \\
        soc-delicious       & 536,108    & 1,365,961  & 293.3 & 412.0 & 610.8 & {232.4} & {77.47} & {130.0} & 1.675 & 1.670 & 1.705 & 1.684 & 1.682 & 1.681 \\
        soc-youtube-snap    & 1,134,890  & 2,987,624  & 417.6 & 613.9 & 896.6  & {184.6} & {62.25} & {101.5} & 1.600 & 1.482 & 1.605 & 1.590 & 1.581 & 1.574 \\
        soc-hyves           & 1,402,673  & 2,777,419  & 370.5 & 514.8 & 801.6 & {245.4} & {78.71} & {131.5} & 1.555 & 1.581 & 1.598 & 1.679 & 1.671 & 1.665 \\
        delaunay-n24        & 16,777,216 & 50,331,601 & 232.6 & 281.3 & 375.0 & {9.779} & {15.64} & {28.22} & 1.703 & 1.698 & 1.699 & 1.696 & 1.699 & 1.699 \\
        road-usa            & 23,947,347 & 28,854,312 & 3470  & 4708  & --    & 30.53 & 46.60 & 81.73 & 4.932 & 4.843 & --    & 4.895 & 4.891 & 4.902 \\
        socfb-konect        & 58,790,782 & 92,208,195 & --    & --    & --    & 1638  & 2354  & 4338  & --    & --    & --    & 2.237 & 2.256 & 2.254 \\
        \bottomrule
    \end{tabular}
\end{table*}

\subsection{Results on Real-World Networks}

\begin{envblue}
    We evaluate the computational efficiency and scalability of our algorithms against exact computation and the baseline \textsc{SimulateMC} under various network sizes and error parameters \(\epsilon\in \left\{0.35,0.3,0.25\right\}\). Tables~\ref{tab:real_disagree} and~\ref{tab:large_disagree} demonstrate our algorithms' consistent speed superiority across all tested networks.
\end{envblue}
Note that one can exactly compute opinion disagreement only for small networks with less than 60,000 nodes, due to its $O(N^3)$ time complexity. Meanwhile, the time complexity of the baseline \textsc{SimulateMC} is \(O(N^2)\), as it needs to estimate hitting times of all node pairs in the network. For networks with more than 100,000 nodes in Table~\ref{tab:large_disagree}, the exact algorithm and \textsc{SimulateMC} runs out of time. For the network socfb-konect with more than 58 million nodes, \textsc{ApproxDelta} fails to complete within a day, whereas \textsc{SampleDelta} finishes in just over an hour, highlighting its superior efficiency and scalability.


Although \textsc{SampleDelta} is significantly faster than \textsc{ApproxDelta}, both algorithms maintain comparable high accuracy on real-world networks. Table~\ref{tab:real_disagree} compares the mean relative errors of our algorithms against the baseline \textsc{SimulateMC}. Results show that our algorithms yield negligible errors, which are orders of magnitude lower than the theoretical guarantees.  Moreover, the two relative errors are nearly identical. Thus, in addition to high efficiency, \textsc{SampleDelta} provides accurate approximations for real-world networks.

It is worth mentioning that for those large networks in Table~\ref{tab:large_disagree} their approximate opinion disagreement is also small,  most of which is close to $1$. This further verifies our claim that the opinion disagreement in scale-free networks does not increase with the network size $N$.

\begin{table*}
    \tabcolsep=4pt
    \caption{The running time (seconds) and mean relative error of \textsc{ApproxDelta} and \textsc{SampleDelta} with various \(\epsilon\) for estimating the Kemeny constant on model networks $\mathcal{F}'_g$.}
    \label{tab:Kemeny}
    \begin{tabular}{@{}cccccccccccccc@{}}
        \toprule
        \multirow{3}{*}{Network}
                              &
        \multirow{3}{*}{$K(\mathcal{F}'_g)$}
                              &
        \multicolumn{6}{c}{Running time (seconds)}
                              &
        \multicolumn{6}{c}{Mean relative error}
        \\
        \cmidrule(l){3-14}
                              &
                              &
        \multicolumn{3}{c}{\textsc{ApproxDelta}}
                              &
        \multicolumn{3}{c}{\textsc{SampleDelta}}
                              &
        \multicolumn{3}{c}{\textsc{ApproxDelta}}
                              &
        \multicolumn{3}{c}{\textsc{SampleDelta}}
        \\
        \cmidrule(l){3-14}
                              &
                              & $\epsilon=0.35$
                              & $\epsilon=0.3$
                              & $\epsilon=0.25$
                              & $\epsilon=0.35$
                              & $\epsilon=0.3$
                              & $\epsilon=0.25$
                              & $\epsilon=0.35$
                              & $\epsilon=0.3$
                              & $\epsilon=0.25$
                              & $\epsilon=0.35$
                              & $\epsilon=0.3$
                              & $\epsilon=0.25$                                                                                                        \\
        \midrule
        \(\mathcal{F}'_{12}\) & \(1.15 \times 10^{6}\) & 186.6 & 244.2 & 379.3 & {65.76} & {108.2} & {blue}{196.0} & 0.017 & 0.013 & 0.009 & 0.002 & 0.002 & 0.000 \\
        \(\mathcal{F}'_{13}\) & \(3.45 \times 10^{6}\) & 425.7 & 590.2 & 836.9 & 114.0 & 174.1 & 328.8 & 0.026 & 0.019 & 0.013 & 0.002 & 0.001 & 0.001 \\
        \(\mathcal{F}'_{14}\) & \(1.04 \times 10^{7}\) & 1454  & 2064  & 2895  & 276.4 & 459.3 & 761.3 & 0.026 & 0.019 & 0.013 & 0.001 & 0.001 & 0.001 \\
        \(\mathcal{F}'_{15}\) & \(3.11 \times 10^{7}\) & --    & --    & --    & {527.6} & {913.0} & {1610} & --    & --    & --    & 0.002 & 0.001 & 0.001 \\
        \bottomrule
    \end{tabular}
\end{table*}

\subsection{Results on a Model Network}

As shown in~\eqref{eq:delta02}, opinion disagreement $\delta$ of a graph $\G$ is dependent on the partial mean hitting time $H_i(\G')$ on the corresponding  graph  $\G'$. It has been proved~\cite{ZhXuZh20} that for a graph $\G$, the partial mean hitting time $H_i(\G)$ is closely related to its Kemeny constant $K(\G)$~\cite{KeSn76}, defined as the expected time for a walker starting from a node $i$ to another node $j$ chosen randomly from  set $\V$ according to the stationary distribution $\ppi$. Specifically,
$K(\G)=\sum_{j=1}^{N}\pi_jH_j(\G) = \sum_{k=2}^{N}\frac{1}{1-\lambda_{k}}$.
Thus, we have
$K(\G')=\sum_{j=1}^{N}\pi_jH_j(\G') = \sum_{k=2}^{N}\frac{1}{1-\lambda^2_{k}}$,
which can be approximated by \textsc{ApproxDelta} and \textsc{SampleDelta}.

We further validate our algorithms by comparing approximated Kemeny constants against exact results on a model network.
\begin{envblue}
    The error parameter of both our algorithms is also set to be \(\epsilon\in \left\{0.35,0.3,0.25\right\}\).
\end{envblue}
The model network is called the pseudofractal scale-free web (PSFW) constructed in an iterative way~\cite{XiZhCo16}. Let $\mathcal{F}_g$ ($g \geq 0$) be the PSFW after $g$ iterations. For $g=0$, $ \mathcal{F}_0$ consists of three nodes and three edges, which forms a triangle. For $g>0$, $\mathcal{F}_g$ is obtained from $\mathcal{F}_{g-1}$ by performing the following operation for each edge $e$ in $\mathcal{F}_{g-1}$: create a new node and link it to both end nodes of the edge $e$. The construction process of the PSFW is illustrated in Figure~\ref{psfw1}.
In graph $\mathcal{F}_g$, there are $(3^{g+1}+3)/2$ nodes and $3^{g+1}$ edges. Let $\PP_g$ denote the transition matrix for graph $\mathcal{F}_{g}$, the eigenvalues of which have been explicitly determined~\cite{XiZhCo16}.
\begin{lemma}\label{lem:Fg}
    For $g\ge 2$, the set of eigenvalues of the transition matrix $\PP_g$ for $\mathcal{F}_{g}$ consists of $1$ with multiplicity one,
    $1-\frac{3}{2^{s+1}}$ with multiplicity $\frac{3^{g-s}+3}{2}$ for $s=0,1,\ldots,g$,
    and
    $1-\frac{4}{2^{s+2}}$ with multiplicity $\frac{3^{g-s}-3}{2}$ for $s=0,1,\ldots,g-2$.
\end{lemma}

\begin{figure}[htbp]
    \begin{center}
        \includegraphics[width=0.6\linewidth]{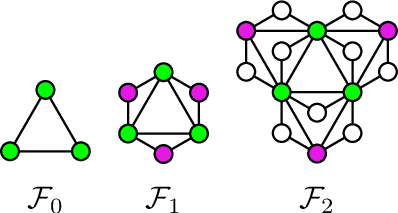}
        \caption{Construction of the first several iterations of the PSFW. }
        \label{psfw1}
    \end{center}
\end{figure}

Let  $\mathcal{F}'_g$ ($g \geq 0$) be the network associated with the PSFW $\mathcal{F}_g$, whose transition matrix is $\PP^2_g$. According to Lemma~\ref{lem:Fg}, the Kemeny constant $K(\mathcal{F}'_g)$ of graph $\mathcal{F}'_g$ can be expressed analytically as
\begin{small}
    \begin{align}
        K(\mathcal{F}'_g)
        =\sum_{s=0}^{g}\frac{2^{2s+1}\cdot\kh{3^{g-s-1}+1}}{2^{s+2}-3}
        +\sum_{s=0}^{g-2}\frac{2^{2s-1}\cdot\kh{3^{g-s}-3}}{2^{s+1}-1}.
        \label{Kemeny04}
    \end{align}
\end{small}

We also use algorithms \textsc{ApproxDelta} and \textsc{SampleDelta} to estimate the Kemeny constant $K(\mathcal{F}'_g)$ on four large networks: $\mathcal{F}'_{12}$, $\mathcal{F}'_{13}$, $\mathcal{F}'_{14}$ and $\mathcal{F}'_{15}$, with $\mathcal{F}'_{15}$ having 21,523,362 nodes.
The mean relative error for \textsc{ApproxDelta} and \textsc{SampleDelta} are reported in Table~\ref{tab:Kemeny}, which shows our algorithms \textsc{ApproxDelta} and \textsc{SampleDelta} always achieve good performance for graphs $\mathcal{F}'_{12}$ and $\mathcal{F}'_{13}$, but \textsc{SampleDelta} is a little effective and much efficient than \textsc{ApproxDelta}. For graphs $\mathcal{F}'_{15}$, \textsc{ApproxDelta} runs out of time and memory. However, \textsc{SampleDelta} returns results within 920 seconds, with a relative error smaller than $0.3\%$.


\subsection{Influence of Varying Error Parameter}

\begin{envblue}
    Through comprehensive evaluation of algorithmic performance, we observe substantial parameter dependence on the error parameter \(\epsilon\). To quantify this relationship, we conduct controlled experiments across diverse network scales by systematically adjusting \(\epsilon\) within the range \([0.25,0.4]\). Our investigation employs two critical measures: computational efficiency across varying network sizes and approximation accuracy of estimating Kemeny constant.
\end{envblue}

\begin{figure}[htbp]
    \centering
    \includegraphics[width=0.455\textwidth]{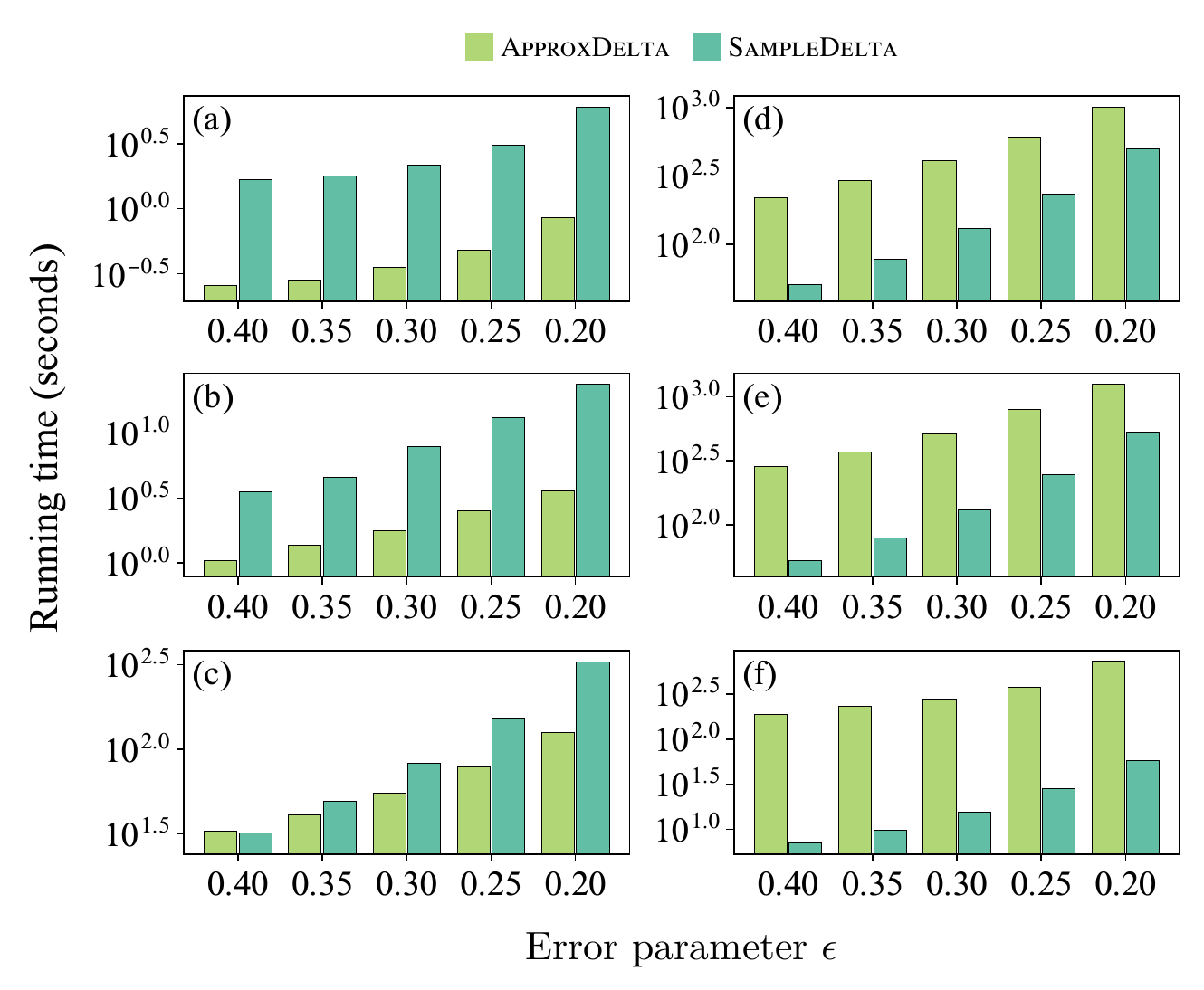}
    \caption{Running time of different algorithms with varying error parameter \(\epsilon\) on real-world graphs: Protein (a), web-EPA (b), Brightkite (c), soc-delicious (d), soc-hyves (e), and delaunay-n24 (f).}
    \label{fig:realistic_time_bar}
\end{figure}

\subsubsection{Effect on Efficiency}
\begin{envblue}
    Figure~\ref{fig:realistic_time_bar} illustrates the computational complexity characteristics of our methods under different \(\epsilon\) configurations. Both algorithms demonstrate runtime patterns that correspond to their theoretical complexity factors of \(\epsilon^{-2}\), though with distinct scaling behaviors across network sizes. 
    On smaller networks like Protein and web-EPA, \textsc{ApproxDelta} excels due to optimized matrix operations. However, \textsc{SampleDelta} gains substantial efficiency advantages on massive networks like soc-hyves and delaunay-n24 through its sublinear sampling approach. This performance inversion stems from the fundamental complexity characteristics. \textsc{SampleDelta}'s runtime improvement scales with \(\sqrt{n}\), making its sampling overhead only beneficial beyond certain network size thresholds. While both methods dramatically outperform conventional approaches, \textsc{SampleDelta} demonstrates superior scalability for networks with over one million nodes.
\end{envblue}

\begin{figure}[htbp]
    \centering
    \includegraphics[width=0.455\textwidth]{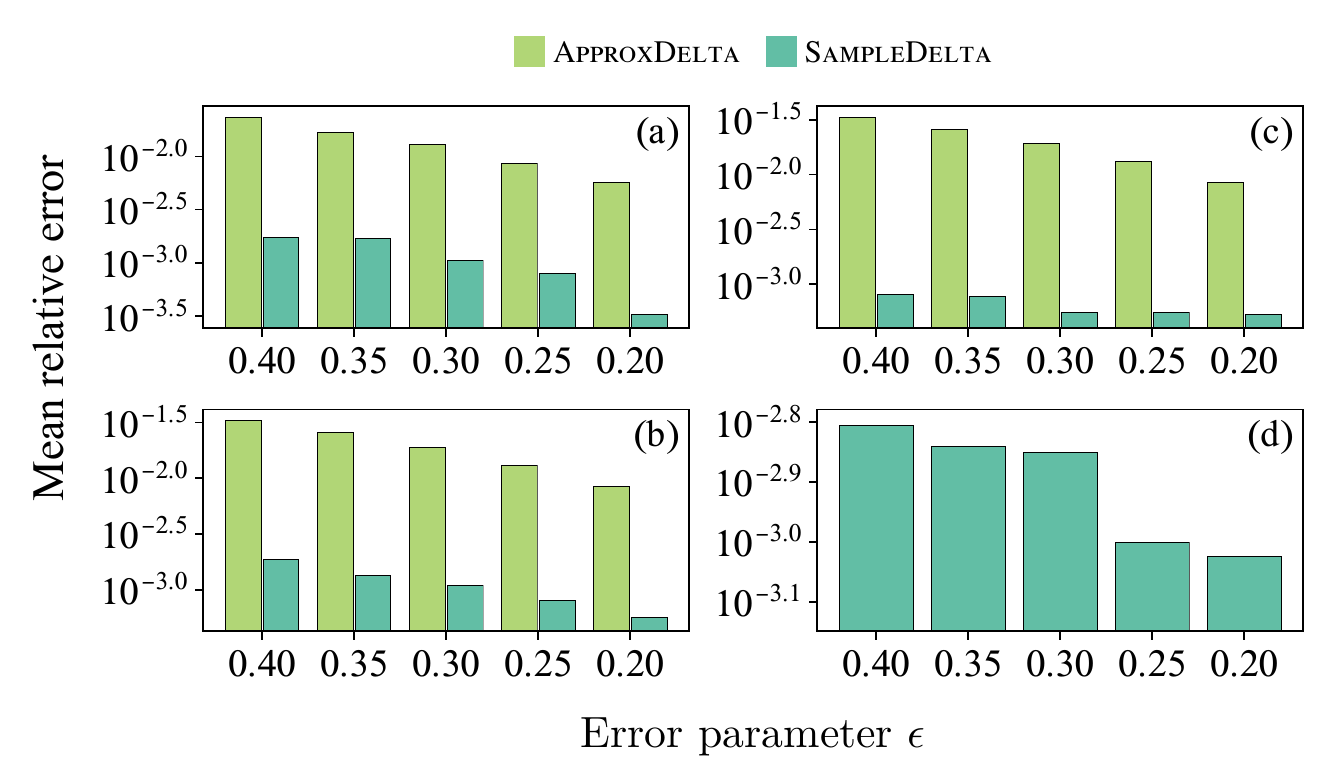}
    \caption{Mean relative error of different algorithms with varying error parameter \(\epsilon\) on the pseudofractal scale-free webs: \(\mathcal{F}'_{12}\), \(\mathcal{F}'_{13}\), \(\mathcal{F}'_{14}\), and \(\mathcal{F}'_{15}\).}
    \label{fig:model_error_bar}
\end{figure}

\subsubsection{Effect on Accuracy}
\begin{envblue}
    Figure~\ref{fig:model_error_bar} reveals the accuracy characteristics of our algorithms when estimating Kemeny constants for the PSFWs. As displayed in Figure~\ref{fig:model_error_bar}, \textsc{SampleDelta} maintains consistently higher accuracy across all \(\epsilon\) values tested. While larger values of \(\epsilon\) produce elevated relative errors, reducing the parameter to \(0.3\) or below yields estimation errors within acceptable bounds. This precision thresholding demonstrates effective error control mechanisms in both algorithms, with accuracy improvements plateauing when \(\epsilon < 0.3\). The observed error profiles confirm that parameter selection directly governs the fundamental trade-off between computational efficiency and estimation accuracy.
\end{envblue}

\section{Conclusions}
In this paper, we presented a study of the noisy DeGroot model on power-law graphs, with an aim to explore the impact of scale-free structure on the limiting opinions. The introduction of noise prevents the opinions of agents from reaching agreement, which fluctuate around their weighted average after long-time evaluation, leading to opinion disagreement or opinion diversity, which can be measured by a quantitative index, that is, expected squared deviation.

The main works of this paper are as follows. First, we studied the opinion disagreement for some representative real-world scale-free networks, BA networks, and random Apollonian networks, and found that opinion disagreement is significantly small, independent of the network size, which shows that the impact of noise on the opinion dynamics in power-law graphs is negligible, due to the scale-free structure. Then, we developed a sampling algorithm to approximate the opinion disagreement, which has a sublinear time complexity with respect to the number of nodes. Finally, we validated the efficiency and accuracy of our algorithm by extensive experiments on a large variety of realistic and model networks. 
While our current model relies on idealized noise assumptions, future work will focus on enhancing the framework's robustness against complex environmental dynamics and potential adversarial behaviors.



\bibliographystyle{IEEEtran}

\begin{thebibliography}{10}
\providecommand{\url}[1]{#1}
\csname url@samestyle\endcsname
\providecommand{\newblock}{\relax}
\providecommand{\bibinfo}[2]{#2}
\providecommand{\BIBentrySTDinterwordspacing}{\spaceskip=0pt\relax}
\providecommand{\BIBentryALTinterwordstretchfactor}{4}
\providecommand{\BIBentryALTinterwordspacing}{\spaceskip=\fontdimen2\font plus
\BIBentryALTinterwordstretchfactor\fontdimen3\font minus \fontdimen4\font\relax}
\providecommand{\BIBforeignlanguage}[2]{{%
\expandafter\ifx\csname l@#1\endcsname\relax
\typeout{** WARNING: IEEEtran.bst: No hyphenation pattern has been}%
\typeout{** loaded for the language `#1'. Using the pattern for}%
\typeout{** the default language instead.}%
\else
\language=\csname l@#1\endcsname
\fi
#2}}
\providecommand{\BIBdecl}{\relax}
\BIBdecl

\bibitem{Le20}
H.~Ledford, ``How facebook, twitter and other data troves are revolutionizing social science,'' \emph{Nature}, vol. 582, no. 7812, pp. 328--330, 2020.

\bibitem{SmCh08}
K.~P. Smith and N.~A. Christakis, ``Social networks and health,'' \emph{Annu. Rev. Sociol.}, vol.~34, no.~1, pp. 405--429, 2008.

\bibitem{AnYe19}
B.~D. Anderson and M.~Ye, ``Recent advances in the modelling and analysis of opinion dynamics on influence networks,'' \emph{Int. J. Autom. Comput.}, vol.~16, no.~2, pp. 129--149, 2019.

\bibitem{MaTeTs17}
A.~Matakos, E.~Terzi, and P.~Tsaparas, ``Measuring and moderating opinion polarization in social networks,'' \emph{Data Min. Knowl. Discov.}, vol.~31, no.~5, pp. 1480--1505, 2017.

\bibitem{MuMuTs18}
C.~Musco, C.~Musco, and C.~E. Tsourakakis, ``Minimizing polarization and disagreement in social networks,'' in \emph{Proc. 2018 World Wide Web Conference}.\hskip 1em plus 0.5em minus 0.4em\relax International World Wide Web Conferences Steering Committee, 2018, pp. 369--378.

\bibitem{ZhBaZh21}
L.~Zhu, Q.~Bao, and Z.~Zhang, ``Minimizing polarization and disagreement in social networks via link recommendation,'' \emph{Advances in Neural Information Processing Systems}, vol.~34, pp. 2072--2084, 2021.

\bibitem{HaMeRiUp21}
S.~Haddadan, C.~Menghini, M.~Riondato, and E.~Upfal, ``Republik: Reducing polarized bubble radius with link insertions,'' in \emph{Proc. 14th ACM International Conference on Web Search and Data Mining}, 2021, pp. 139--147.

\bibitem{GaKlTa20}
J.~Gaitonde, J.~Kleinberg, and E.~Tardos, ``Adversarial perturbations of opinion dynamics in networks,'' in \emph{Proc. 21st ACM Conference on Economics and Computation}, 2020, pp. 471--472.

\bibitem{WaZhZh25}
G.~Wang, R.~Zhang, and Z.~Zhang, ``Efficient algorithms for relevant quantities of {F}riedkin-{J}ohnsen opinion dynamics model,'' in \emph{Proc. 31st ACM SIGKDD International Conference on Knowledge Discovery \& Data Mining}.\hskip 1em plus 0.5em minus 0.4em\relax ACM, 2025, pp. 2915--2926.

\bibitem{SuSuZhZh25}
H.~Sun, Y.~Sun, X.~Zhou, and Z.~Zhang, ``Fast computation and optimization for opinion-based quantities of {F}riedkin-{J}ohnsen model,'' in \emph{The Thirty-ninth Annual Conference on Neural Information Processing Systems}, 2025.

\bibitem{QiQiZhLi26}
X.~Qian, B.-W. Qin, H.~Zhu, and W.~Lin, ``Opinion polarization and its connected disagreement: Modeling and modulation,'' \emph{Phys. Rev. E}, vol. 113, no.~1, p. L012301, 2026.

\bibitem{MaPa19}
E.~Mackin and S.~Patterson, ``Maximizing diversity of opinion in social networks,'' in \emph{Proc. 2019 Amer. Control Conf.}\hskip 1em plus 0.5em minus 0.4em\relax IEEE, 2019, pp. 2728--2734.

\bibitem{AxDaFo21}
R.~Axelrod, J.~J. Daymude, and S.~Forrest, ``Preventing extreme polarization of political attitudes,'' \emph{Proc. Natl. Acad. Sci. U.S.A.}, vol. 118, no.~50, p. e2102139118, 2021.

\bibitem{YiPa20}
Y.~Yi and S.~Patterson, ``Disagreement and polarization in two-party social networks,'' \emph{IFAC-PapersOnLine}, vol.~53, no.~2, pp. 2568--2575, 2020.

\bibitem{XuBaZh21}
W.~Xu, Q.~Bao, and Z.~Zhang, ``Fast evaluation for relevant quantities of opinion dynamics,'' in \emph{Proc. the Web Conference}, 2021, pp. 2037--2045.

\bibitem{SaKeKhLa23}
A.~Saha, X.~Ke, A.~Khan, and L.~V. Lakshmanan, ``Voting-based opinion maximization,'' in \emph{2023 IEEE 39th International Conference on Data Engineering}.\hskip 1em plus 0.5em minus 0.4em\relax IEEE, 2023, pp. 544--557.

\bibitem{No20}
H.~Noorazar, ``Recent advances in opinion propagation dynamics: {A} 2020 survey,'' \emph{Eur. Phys. J. Plus}, vol. 135, no.~6, p. 521, 2020.

\bibitem{BeWaVaHoShAl21}
C.~Bernardo, L.~Wang, F.~Vasca, Y.~Hong, G.~Shi, and C.~Altafini, ``Achieving consensus in multilateral international negotiations: The case study of the 2015 paris agreement on climate change,'' \emph{Sci. Adv.}, vol.~7, no.~51, p. eabg8068, 2021.

\bibitem{FrPrTePa16}
N.~E. Friedkin, A.~V. Proskurnikov, R.~Tempo, and S.~E. Parsegov, ``Network science on belief system dynamics under logic constraints,'' \emph{Science}, vol. 354, no. 6310, pp. 321--326, 2016.

\bibitem{FrBu17}
N.~E. Friedkin and F.~Bullo, ``How truth wins in opinion dynamics along issue sequences,'' \emph{Proc. Natl. Acad. Sci. U.S.A.}, vol. 114, no.~43, pp. 11\,380--11\,385, 2017.

\bibitem{De74}
M.~H. Degroot, ``Reaching a consensus,'' \emph{J. Am. Statist. Assoc.}, vol.~69, no. 345, pp. 118--121, 1974.

\bibitem{XiBoKi07}
L.~Xiao, S.~Boyd, and S.-J. Kim, ``Distributed average consensus with least-mean-square deviation,'' \emph{J. Parallel. Distrib. Comput.}, vol.~67, no.~1, pp. 33--46, 2007.

\bibitem{JaOl19}
A.~Jadbabaie and A.~Olshevsky, ``Scaling laws for consensus protocols subject to noise,'' \emph{IEEE Trans. Automat. Contr.}, vol.~64, no.~4, pp. 1389--1402, 2019.

\bibitem{XuZh23CIKM}
W.~Xu and Z.~Zhang, ``Minimizing polarization in noisy leader-follower opinion dynamics,'' in \emph{Proc. 32nd ACM International Conference on Information and Knowledge Management}, 2023, pp. 2856--2865.

\bibitem{BaAl99}
A.-L. Barab{\'a}si and R.~Albert, ``Emergence of scaling in random networks,'' \emph{Science}, vol. 286, no. 5439, pp. 509--512, 1999.

\bibitem{WeHeXiWaLiDuWe19}
Z.~Wei, X.~He, X.~Xiao, S.~Wang, Y.~Liu, X.~Du, and J.-R. Wen, ``{PRSim}: Sublinear time simrank computation on large power-law graphs,'' in \emph{Proc. 2019 International Conference on Management of Data}, 2019, pp. 1042--1059.

\bibitem{XuShZhKaZh20}
W.~Xu, Y.~Sheng, Z.~Zhang, H.~Kan, and Z.~Zhang, ``Power-law graphs have minimal scaling of {K}emeny constant for random walks,'' in \emph{Proc. the Web Conference 2020}, 2020, pp. 46--56.

\bibitem{MaJa21}
R.~Mayer and H.-A. Jacobsen, ``Hybrid edge partitioner: Partitioning large power-law graphs under memory constraints,'' in \emph{Proc. 2021 International Conference on Management of Data}, 2021, pp. 1289--1302.

\bibitem{XuZh23}
W.~Xu and Z.~Zhang, ``Optimal scale-free small-world graphs with minimum scaling of cover time,'' \emph{ACM Trans. Knowl. Discov. Data}, vol.~17, no.~7, p.~93, 2023.

\bibitem{Ne03}
M.~E.~J. Newman, ``The structure and function of complex networks,'' \emph{SIAM Rev.}, vol.~45, no.~2, pp. 167--256, 2003.

\bibitem{Be81}
R.~L. Berger, ``{A necessary and sufficient condition for reaching a consensus using DeGroot's method},'' \emph{J. Am. Stat. Assoc.}, vol.~76, no. 374, pp. 415--418, 1981.

\bibitem{FrJo90}
N.~E. Friedkin and E.~C. Johnsen, ``Social influence and opinions,'' \emph{J. Math. Sociol.}, vol.~15, no. 3-4, pp. 193--206, 1990.

\bibitem{ChMu20}
U.~Chitra and C.~Musco, ``Analyzing the impact of filter bubbles on social network polarization,'' in \emph{Proc. 13th International Conference on Web Search and Data Mining}, 2020, pp. 115--123.

\bibitem{TuNe22}
S.~Tu and S.~Neumann, ``A viral marketing-based model for opinion dynamics in online social networks,'' in \emph{Proc. ACM Web Conference 2022}, 2022, pp. 1570--1578.

\bibitem{Al13}
C.~Altafini, ``Consensus problems on networks with antagonistic interactions,'' \emph{IEEE Trans. Autom. Control}, vol.~58, no.~4, pp. 935--946, 2013.

\bibitem{ZhSuXuLiZh24}
X.~Zhou, H.~Sun, W.~Xu, W.~Li, and Z.~Zhang, ``{Friedkin-Johnsen model for opinion dynamics on signed graphs},'' \emph{IEEE Trans. Knowl. Data Eng.}, vol.~36, no.~12, pp. 8313--8327, 2024.

\bibitem{ZhWuAl20}
Q.~Zhou, Z.~Wu, A.~H. Altalhi, and F.~Herrera, ``A two-step communication opinion dynamics model with self-persistence and influence index for social networks based on the {DeGroot} model,'' \emph{Inf. Sci.}, vol. 519, pp. 363--381, 2020.

\bibitem{ZhXuZhCh24}
Z.~Zhang, W.~Xu, Z.~Zhang, and G.~Chen, ``Opinion dynamics in social networks incorporating higher-order interactions,'' \emph{Data Mining and Knowledge Discovery}, vol.~38, no.~6, pp. 4001--4023, 2024.

\bibitem{MeAsDaAmAn13}
E.~Yildiz, A.~Ozdaglar, D.~Acemoglu, A.~Saberi, and A.~Scaglione, ``Binary opinion dynamics with stubborn agents,'' \emph{ACM Trans. Econ. Comput.}, vol.~1, no.~4, pp. 1--30, 2013.

\bibitem{VaFaFr14}
V.~Luca, F.~Fabio, F.~Paolo, and O.~Asuman, ``Message passing optimization of harmonic influence centrality,'' \emph{IEEE Trans. Control Netw. Syst.}, vol.~1, no.~1, pp. 109--120, 2014.

\bibitem{MaAb19}
V.~S. Mai and E.~H. Abed, ``Optimizing leader influence in networks through selection of direct followers,'' \emph{IEEE Trans. Autom. Control}, vol.~64, no.~3, pp. 1280--1287, 2019.

\bibitem{ZhZh21}
X.~Zhou and Z.~Zhang, ``Maximizing influence of leaders in social networks,'' in \emph{Proc. 27th ACM SIGKDD International Conference on Knowledge Discovery \& Data Mining}.\hskip 1em plus 0.5em minus 0.4em\relax ACM, 2021, pp. 2400--2408.

\bibitem{ZhZhLiZh23}
X.~Zhou, L.~Zhu, W.~Li, and Z.~Zhang, ``A sublinear time algorithm for opinion optimization in directed social networks via edge recommendation,'' in \emph{Proc. 29th ACM SIGKDD International Conference on Knowledge Discovery \& Data Mining}.\hskip 1em plus 0.5em minus 0.4em\relax ACM, 2023, pp. 3593--3602.

\bibitem{ZhZh23}
X.~Zhou and Z.~Zhang, ``Opinion maximization in social networks via leader selection,'' in \emph{Proc. the ACM Web Conference}, 2023, pp. 133--142.

\bibitem{DaGoLe13}
P.~Dandekar, A.~Goel, and D.~T. Lee, ``Biased assimilation, homophily, and the dynamics of polarization,'' \emph{Proc. Natl. Acad. Sci. U.S.A.}, vol. 110, no.~15, pp. 5791--5796, 2013.

\bibitem{DoDiMa17}
Y.~Dong, Z.~Ding, L.~Mart{\'\i}nez, and F.~Herrera, ``Managing consensus based on leadership in opinion dynamics,'' \emph{Inf. Sci.}, vol. 397, pp. 187--205, 2017.

\bibitem{StLi20}
S.~Stern and G.~Livan, ``The impact of noise and topology on opinion dynamics in social networks,'' \emph{R. Soc. Open Sci.}, vol.~8, no. 201943, 2021.

\bibitem{LoGaZa13}
E.~Lovisari, F.~Garin, and S.~Zampieri, ``Resistance-based performance analysis of the consensus algorithm over geometric graphs,'' \emph{SIAM J. Control Optim.}, vol.~51, no.~5, pp. 3918--3945, 2013.

\bibitem{So18}
P.~Sobkowicz, ``Opinion dynamics model based on cognitive biases of complex agents,'' \emph{Journal of Artificial Societies and Social Simulation}, vol.~21, no.~4, p.~8, 2018.

\bibitem{ChWaTs22}
T.~Chen, X.~Wang, and C.~E. Tsourakakis, ``Polarizing opinion dynamics with confirmation bias,'' in \emph{Social Informatics}, 2022, pp. 144--158.

\bibitem{BeVaIe22}
C.~Bernardo, F.~Vasca, and R.~Iervolino, ``Heterogeneous opinion dynamics with confidence thresholds adaptation,'' \emph{IEEE Transactions on Control of Network Systems}, vol.~9, no.~3, pp. 1068--1079, 2022.

\bibitem{MeRaBrHaRo24}
V.~Mengers, M.~Raoufi, O.~Brock, H.~Hamann, and P.~Romanczuk, ``Leveraging uncertainty in collective opinion dynamics with heterogeneity,'' \emph{Sci. Rep.}, vol.~14, no.~1, p. 27314, 2024.

\bibitem{AbKlPaTs18}
R.~Abebe, J.~Kleinberg, D.~Parkes, and C.~E. Tsourakakis, ``Opinion dynamics with varying susceptibility to persuasion,'' in \emph{Proc. 24th ACM SIGKDD International Conference on Knowledge Discovery \& Data Mining}, 2018, p. 1089–1098.

\bibitem{LiYeAnBaNe18}
J.~Liu, M.~Ye, B.~D. Anderson, T.~Basar, and A.~Nedic, ``Discrete-time polar opinion dynamics with heterogeneous individuals,'' in \emph{2018 IEEE Conference on Decision and Control (CDC)}, 2018, pp. 1694--1699.

\bibitem{Ch18}
J.-H. Cho, ``Dynamics of uncertain and conflicting opinions in social networks,'' \emph{IEEE Trans. Comput. Soc. Syst.}, vol.~5, no.~2, pp. 518--531, 2018.

\bibitem{ZhLiKoDoYu19}
M.~Zhan, H.~Liang, G.~Kou, Y.~Dong, and S.~Yu, ``Impact of social network structures on uncertain opinion formation,'' \emph{IEEE Trans. Comput. Soc. Syst.}, vol.~6, no.~4, pp. 670--679, 2019.

\bibitem{FoKaKoSk18}
D.~Fotakis, V.~Kandiros, V.~Kontonis, and S.~Skoulakis, ``Opinion dynamics with limited information,'' in \emph{Web and Internet Economics}, 2018, pp. 282--296.

\bibitem{HeFaZhWa23}
Q.~He, H.~Fang, J.~Zhang, and X.~Wang, ``Dynamic opinion maximization in social networks,'' \emph{IEEE Trans. Knowl. Data Eng.}, vol.~35, no.~1, pp. 350--361, 2023.

\bibitem{ArPaSa23}
A.~Arya, P.~K. Pandey, and A.~Saxena, ``Balanced and unbalanced triangle count in signed networks,'' \emph{IEEE Trans. Knowl. Data Eng.}, vol.~35, no.~12, pp. 12\,491--12\,496, 2023.

\bibitem{HeFeChLiHuTa22}
C.~He, X.~Fei, Q.~Cheng, H.~Li, Z.~Hu, and Y.~Tang, ``A survey of community detection in complex networks using nonnegative matrix factorization,'' \emph{IEEE Trans. Comput. Soc. Syst.}, vol.~9, no.~2, pp. 440--457, 2022.

\bibitem{ArPa24}
A.~Arya and P.~K. Pandey, ``Sissrm: Sequentially induced signed subnetwork reconstruction model for generating realistic synthetic signed networks,'' \emph{IEEE Trans. Comput. Soc. Syst.}, vol.~11, no.~5, pp. 6476--6486, 2024.

\bibitem{HeSuWaWaHuYiWaMa21}
Q.~He, L.~Sun, X.~Wang, Z.~Wang, M.~Huang, B.~Yi, Y.~Wang, and L.~Ma, ``Positive opinion maximization in signed social networks,'' \emph{Inf. Sci.}, vol. 558, pp. 34--49, 2021.

\bibitem{ChLiDe18}
X.~Chen, J.~Lijffijt, and T.~De~Bie, ``Quantifying and minimizing risk of conflict in social networks,'' in \emph{Proc. 24th ACM SIGKDD International Conference on Knowledge Discovery \& Data Mining}, 2018, p. 1197–1205.

\bibitem{Ch97}
F.~R. Chung, \emph{Spectral Graph Theory}.\hskip 1em plus 0.5em minus 0.4em\relax American Mathematical Society, Providence, RI, 1997.

\bibitem{KeSn76}
J.~G. Kemeny and J.~L. Snell, \emph{Finite Markov Chains}.\hskip 1em plus 0.5em minus 0.4em\relax Springer, New York, 1976.

\bibitem{Lo93}
L.~Lov{\'a}sz, ``{Random walks on graphs: A survey},'' \emph{Combinatorics, Paul Erd\"{o}s is eighty}, vol.~2, no.~1, pp. 1--46, 1993.

\bibitem{CoBeTeVoKl07}
S.~Condamin, O.~B{\'e}nichou, V.~Tejedor, R.~Voituriez, and J.~Klafter, ``First-passage times in complex scale-invariant media,'' \emph{Nature}, vol. 450, no. 7166, pp. 77--80, 2007.

\bibitem{Be09}
A.~Beveridge, ``Centers for random walks on trees,'' \emph{SIAM J. Discrete Math.}, vol.~23, no.~1, pp. 300--318, 2009.

\bibitem{TeBeVo09}
V.~Tejedor, O.~B{\'e}nichou, and R.~Voituriez, ``Global mean first-passage times of random walks on complex networks,'' \emph{Phys. Rev. E}, vol.~80, no.~6, p. 065104, 2009.

\bibitem{MaMagi15}
C.~Mavroforakis, M.~Mathioudakis, and A.~Gionis, ``{Absorbing random-walk centrality: Theory and algorithms},'' in \emph{Proc. IEEE International Conference on Data Mining}.\hskip 1em plus 0.5em minus 0.4em\relax IEEE, 2015, pp. 901--906.

\bibitem{ZhXuZh20}
Z.~Zhang, W.~Xu, and Z.~Zhang, ``Nearly linear time algorithm for mean hitting times of random walks on a graph,'' in \emph{Proc. 13th International Conference on Web Search and Data Mining}, 2020, pp. 726--734.

\bibitem{XiXuZhZh25}
H.~Xia, W.~Xu, Z.~Zhang, and Z.~Zhang, ``Means of hitting times for random walks on graphs: {C}onnections, computation, and optimization,'' \emph{ACM Trans. Knowl. Discov. Data}, vol.~19, no.~2, pp. 1--35, 2025.

\bibitem{ChChLiPeTe15b}
D.~Cheng, Y.~Cheng, Y.~Liu, R.~Peng, and S.-H. Teng, ``{Efficient sampling for Gaussian graphical models via spectral sparsification},'' in \emph{Proc. 28th Conference on Learning Theory}, 2015, pp. 364--390.

\bibitem{BeGrTh74}
A.~Ben-Israel and T.~N.~E. Greville, \emph{Generalized inverses: theory and applications}.\hskip 1em plus 0.5em minus 0.4em\relax J. Wiley, 1974.

\bibitem{Ku13}
J.~Kunegis, ``Konect: The koblenz network collection,'' in \emph{Proc. 22nd International Conference on World Wide Web}, 2013, pp. 1343--1350.

\bibitem{RoAh15}
R.~Rossi and N.~Ahmed, ``The network data repository with interactive graph analytics and visualization,'' in \emph{Proc. 29th AAAI Conference on Artificial Intelligence}.\hskip 1em plus 0.5em minus 0.4em\relax AAAI, 2015, pp. 4292--4293.

\bibitem{ZhRoFr06}
Z.~Zhang, L.~Rong, and F.~Comellas, ``High-dimensional random {A}pollonian networks,'' \emph{Physica A}, vol. 364, pp. 610--618, 2006.

\bibitem{ZhZhShGu07}
Z.~Zhang, S.~Zhou, Z.~Shen, and J.~Guan, ``From regular to growing small-world networks,'' \emph{Physica A: Statistical Mechanics and its Applications}, vol. 385, no.~2, pp. 765--772, 2007.

\bibitem{WaSt98}
D.~J. Watts and S.~H. Strogatz, ``Collective dynamics of `small-world' networks,'' \emph{Nature}, vol. 393, no. 6684, pp. 440--442, 1998.

\bibitem{ShZh19}
Y.~Sheng and Z.~Zhang, ``Low-mean hitting time for random walks on heterogeneous networks,'' \emph{IEEE Trans. Inf. Theory}, vol.~65, no.~11, pp. 6898--6910, 2019.

\bibitem{Bo13}
E.~Bozzo, ``{The Moore--Penrose inverse of the normalized graph Laplacian},'' \emph{Linear Algebra Appl.}, vol. 439, no.~10, pp. 3038--3043, 2013.

\bibitem{GhBoSa08}
A.~Ghosh, S.~Boyd, and A.~Saberi, ``Minimizing effective resistance of a graph,'' \emph{SIAM Rev.}, vol.~50, no.~1, pp. 37--66, 2008.

\bibitem{JoLi84}
W.~B. Johnson and J.~Lindenstrauss, ``{Extensions of Lipschitz mappings into a Hilbert space},'' \emph{Contemp. Math.}, vol.~26, pp. 189--206, 1984.

\bibitem{KoMiPe11}
I.~Koutis, G.~L. Miller, and R.~Peng, ``A nearly-$m$ log $n$ time solver for {SDD} linear systems,'' in \emph{2011 IEEE 52nd Annual Symposium on Foundations of Computer Science}.\hskip 1em plus 0.5em minus 0.4em\relax IEEE, 2011, pp. 590--598.

\bibitem{CoKyMiPaPeRaSu14}
M.~B. Cohen, R.~Kyng, G.~L. Miller, J.~W. Pachocki, R.~Peng, A.~B. Rao, and S.~C. Xu, ``Solving {SDD} linear systems in nearly $m$ log 1/2 $n$ time,'' in \emph{Proc. 46th annual ACM symposium on Theory of computing}.\hskip 1em plus 0.5em minus 0.4em\relax ACM, 2014, pp. 343--352.

\bibitem{Ho63}
W.~Hoeffding, ``Probability inequalities for sums of bounded random variables,'' \emph{J. Amer. Statist. Assoc.}, vol.~58, no. 301, pp. 13--30, 1963.

\bibitem{KySa16}
R.~Kyng and S.~Sachdeva, ``Approximate {G}aussian elimination for {L}aplacians - fast, sparse, and simple,'' in \emph{Proc. {IEEE} 57th Annual Symposium on Foundations of Computer Science}.\hskip 1em plus 0.5em minus 0.4em\relax IEEE, 2016, pp. 573--582.

\bibitem{MoYuKi10}
A.~Mohaisen, A.~Yun, and Y.~Kim, ``Measuring the mixing time of social graphs,'' in \emph{Proc. 10th ACM SIGCOMM Conference on Internet Measurement}, 2010, pp. 383--389.

\bibitem{XiZhCo16}
P.~Xie, Z.~Zhang, and F.~Comellas, ``{On the spectrum of the normalized {L}aplacian of iterated triangulations of graphs},'' \emph{Appl. Math. Comput.}, vol. 273, pp. 1123--1129, 2016.

\end{thebibliography}

\begin{IEEEbiographynophoto} 
    {Wanyue Xu}
    (S'20)	received the B.Eng. degree in computer science and technology, Shangdong University, China, in 2019 and the Ph.D degree in School of Computer Science, Fudan University, China in 2024. She is currently an assistant professor with the School of Business and Management, Shanghai International Studies University,  Shanghai, China. 
    Her research interests include network science, graph data mining, social network analysis, and random walks. 
\end{IEEEbiographynophoto}

\begin{IEEEbiographynophoto}{Yubo Sun}
received the B.Eng. degree in Software
Engineering Institute, East China Normal University, Shanghai, China, in 2023. He is currently pursuing the Ph.D. degree in Research Institute of Intelligent
Complex Systems, Fudan University, Shanghai,
China. His research interests include complex
networks, graph theory, graph algorithms and graph learning.
\end{IEEEbiographynophoto}

\begin{IEEEbiographynophoto}{Mingzhe Zhu}
    received the B.Sc. degree in 2021 and  the Master degree in 2024, both in the
    School of Computer Science, Fudan University, Shanghai, China. His research interests include network science, graph
    data mining, and random walks.
\end{IEEEbiographynophoto}

\begin{IEEEbiographynophoto}
    {Zuobai Zhang} received the B.Sc. degree in the
    School of Computer Science, Fudan University,
    Shanghai, China, in 2021. He is currently pursuing the doctoral degree in Mila-Quebec AI Institute, Canada. His research interests include graph algorithms, social networks, and network science.\\
    Mr. Zhang has published several papers in international journals or conferences, including WSDM, WWW and IEEE Transactions on Cybernetics. He won a Silver Medal in National Olympiad in Informatics of China in 2016 and several Gold Medals in ICPC Asia Regional Contests.
\end{IEEEbiographynophoto}

\begin{IEEEbiographynophoto} 
    {Zhongzhi Zhang}
    (M'19)	 received the B.Sc. degree in applied mathematics from Anhui University, Hefei, China, in 1997 and the Ph.D. degree in management science and engineering from Dalian University of Technology, Dalian, China, in 2006. \\
    From 2006 to 2008, he was a Post-Doctoral Research Fellow with Fudan University, Shanghai, China, where he is currently a Full Professor with the College of Computer Science and Artificial Intelligence. 
    Since 2019, he has been selected as one of the most cited Chinese researchers 	(Elsevier) every year. 
    His current research interests include network science, graph data mining, social network analysis, computational social science, spectral graph theory, and random walks. \\
    Dr. Zhang was a recipient of the Excellent Doctoral Dissertation Award of Liaoning Province, China, in 2007, the Excellent Post-Doctor Award of Fudan University in 2008, the Shanghai Natural Science Award (third class) in 2013, the Wilkes Award for the best paper published in The Computer Journal in 2019, and the CCF Natural Science Award (second class) in 2022. 
\end{IEEEbiographynophoto}

\end{document}